\documentclass[11pt]{article}
\usepackage{a4}
\usepackage{amsthm}
\usepackage{amssymb}
\usepackage{amsfonts}
\usepackage{amsmath}
\usepackage{tikz}
\usepackage{caption}
\usepackage{subcaption}

\usepackage{color}
\newcommand{\cblue}{\color{blue}}
\newcommand{\cred}{\color{red}}
\newcommand{\cgreen}{\color{green}}

\newtheorem{theorem}{Theorem}
\newtheorem{proposition}{Proposition}
\newtheorem{example}{Example}
\newtheorem{lemma}{Lemma}
\newtheorem{corollary}{Corollary}
\newtheorem{remark}{Remark}
\newtheorem{definition}{Definition}

\newtheorem{assumption}{Assumption}
\newtheorem{LP}{Linear Program}

\newcommand {\E}{\mathbb {E}}
\newcommand {\Prob}{\mathbb {P}}
\newcommand{\R}{\mathbb{R}}
\newcommand{\bbF}{\mathbb{F}}
\newcommand{\bbT}{\mathbb{T}}

\newcommand{\sF}{\mathcal F}
\newcommand{\sG}{\mathcal G}
\newcommand{\sH}{\mathcal H}
\newcommand{\sK}{\mathcal K}
\newcommand{\sL}{\mathcal L}
\newcommand{\sM}{\mathcal M}
\newcommand{\sN}{\mathcal N}
\newcommand{\sP}{\mathcal P}
\newcommand{\sS}{\mathcal S}
\newcommand{\sT}{\mathcal T}
\newcommand{\sX}{\mathcal X}
\newcommand{\sY}{\mathcal Y}

\begin{document}

\title{On the value of being American}

\date{\today}
\author{David Hobson \and Anthony Neuberger\thanks{An earlier version of this article~\cite{Neuberger:07} with a single author circulated under the title `Bounds on the American option'.}}

\maketitle

\begin{abstract}
Abstract: The virtue of an American option is that it can be exercised at any time. This right is particularly valuable when there is model uncertainty. Yet almost all the extensive literature on American options assumes away model uncertainty. This paper quantifies the potential value of this flexibility by identifying the supremum on the price of an American option when no model is imposed on the data, but rather any model is required to be consistent with a family of European call prices. The bound is enforced by a hedging strategy involving these call options which is robust to model error.




\end{abstract}

\section{Introduction}
\label{sec:intro}

American options are valuable because the holder is free to react to information, including information that arrives after the option is acquired. Yet almost all the extensive literature on American options makes the extreme (but classical) assumption that the process driving the price of the underlying asset is known perfectly at the outset. In such a world, the holder can identify the optimal exercise strategy at the outset and can, without loss, commit to follow that strategy. Standard valuation methods do not allow for the possibility that evidence from the forward or options market, or events in the real world that occur after the acquisition of the option might cause the holder to change the model and alter the exercise strategy. The American feature provides some protection against model risk, and model based valuation cannot capture the value of this.

This paper investigates how great this extra value could be by looking at valuations that impose minimal restrictions on the price process.  It focusses specifically on the upper bound on the price of an American option given only the contemporaneous prices of European options on the same asset. The bound is enforced by a semi-static hedging strategy which is identified. It avoids the problems of model mis-specification that plague the standard model-based approach. The results are not entirely free from assumptions. Transaction costs and other frictions are ignored; the risk free interest rate and dividend process are assumed to be non-stochastic. But the paper imposes no restrictions on the set of possible paths of the price process (apart from positivity, which is conventional and could readily be relaxed).

There is a substantial literature on model-independent bounds for exotic options in the presence of known European option prices. Originating with work of Hobson~\cite{Hobson:98} for lookback options, model-free or robust bounds have been identified for barrier options (Brown, Hobson, and Rogers~\cite{BrownHobsonRogers:01}), double no touch options (Cox and Obloj~\cite{CoxObloj:11}), basket options (Hobson, Laurence and  Wang~\cite{HobsonLaurenceWang:05}), variance swaps (Hobson and Klimmek~\cite{HobsonKlimmek:12}), options on variance (Carr and Lee~\cite{CarrLee:10})  and forward start options (Hobson and Neuberger~\cite{HobsonNeuberger:12}).
Kahal\'{e}~\cite{Kahale:13} describes a general approach via convex programming for pricing and hedging European path-dependent claims in the presence of European options, using a set-up which is similar to that in the main part of this paper.
Beiglb\"{o}ck, Henry-Labord\`{e}re and Penkner~\cite{BeiglbockHenryLaborderePenkner:13} and Dolinsky and Soner~\cite{DolinskySoner:13} use arguments from the mass-transportation literature to find bounds on general path-dependent options in the presence of European option prices; they show that the dual problem can be interpreted as a robust hedge, as also does Acciaio et al~\cite{AcciaioBeiglbockPenknerSchachermayer:13}. Hobson~\cite{Hobson:11} provides a survey and relates the problem to the Skorokhod embedding problem.

Most of the existing literature is confined to the pricing of European path-dependent claims\footnote{A rare exception is the paper of Cox and Hoeggerl~\cite{CoxHoeggerl:13}. In this paper the aim is to find consistency conditions on the possible shapes (as a function of strike) of the family of prices of American put options with fixed maturity, given the values of co-maturing European puts.}. The contribution of this paper is to find bounds for American claims, and this presents significant new challenges. The holder of a European path-dependent claim is passive, and cannot influence the payoff of the option. The value of the claim depends on the probability which the model assigns to each path, and in the search for the model which is consistent with the European options data  and for which the path-dependent claim has the highest price, it is sufficient to restrict attention to models where the distribution of future returns depends only on past returns. Any additional information about future returns that arrives is irrelevant because there is no mechanism for the holder to respond. In this case the analysis can be restricted, without loss of generality, to models where the filtration is the natural filtration.

By contrast the arrival of new information does affect the exercise decision for an American claim, and hence affects the value. To find the upper bound on the price of the claim, one must search among a much wider set of models, and consider different specifications for the flow of information. As the following simple example shows there are many models consistent with a given set of European call prices (even when we have a complete, double continuum of option prices in strike and maturity) and within this class of models the value of the American option is maximised when model uncertainty is resolved as early as possible. There is a model in which the asset price is Markovian with respect to its natural filtration, but this model underestimates the value of the American option.

\begin{example}
\label{eg:1.1}
Consider a continuous-time world with a single risky asset (the stock) and a riskless bond. The interest rate is zero. There are European call options trading for every strike and maturity. The marginal distribution of the stock price (under any and every consistent pricing measure\footnote{In this context a consistent pricing measure is any measure under which the stock price is a martingale and model based prices of call options i.e. their expected values, agree with the quoted prices.}) is therefore determined at every horizon; in this example, the distribution is a single mass point at 100 for times up to and including 1 year; beyond 1 year there are three equal mass points at 50, 100 and 150. The simplest model consistent with this is a trinomial. With probability 2/3 the price jumps at one year; if it jumps it is equally likely to go up or down by 50; otherwise the price is constant.

There is a family of models for the stock, of which the trinomial is a special case, which is consistent with the data. In each of these models the stock price is constant except at 1 year, when it may jump up or down by 50. The conditional probability of a jump at time 1, given information at time $t$, is a random variable $Z_t$. $Z = (Z_t)_{0 \leq t \leq 1}$ is a martingale, with $Z_0 = 2/3$. The trinomial is a special case where $Z_t = 2/3$ for $0 \leq t < 1$. The time zero value of any European path-dependent claim is the same for all members of this family of models.

Consider a perpetual American claim which pays $[132/(1.1)^t  -  S_t]^+$ if exercised at time t (where $S_t$ is the stock price at t). It is sub-optimal to exercise the claim any time other than immediately, or after 1 year. Under the trinomial model, the holder exercises immediately and receives 32. (Waiting one year would give an expected value of $\frac{1}{3} 70 + \frac{1}{3} 20 + \frac{1}{3} 0=30$.) Suppose now that $Z$ is the left-continuous martingale that jumps from 2/3 to 0 or 1 immediately after time zero. This models the idea that immediately after purchasing the American claim, the holder learns whether the price will change in one year or not. If $Z_{0+} = 0$, the holder exercises immediately, and receives (132-100)=32. If $Z_{0+}=1$, the holder waits and exercises after the price jump, getting 70 or 0 with equal probability. The value of the American claim in this case is $\frac{1}{3} 32 + \frac{2}{3}\frac{1}{2} 70=34$, not 32.
\end{example}

For a European claim, whether holders of the claim get new information after time 0 about the possibility of a jump is immaterial; there is nothing they can do about it. That is why in seeking to find models and strategies that bound the price of path dependent European options, researchers confine themselves to processes and trading strategies that are defined over paths (eg Definition 1.1 in Acciaio et al~\cite{AcciaioBeiglbockPenknerSchachermayer:13}, and Section 2.3 of Dolinsky and Soner~\cite{DolinskySoner:13}). But with an American claim the situation is more subtle because the exercise decision may be altered by the arrival of new information.

The main theoretical result of the paper (Theorem~\ref{thm:main2}) is that the supremum on the value of an American claim, in the presence of a finite set of European call options, is equal to the cost of the cheapest super-replicating strategy. The space of consistent models and the space of super-replicating strategies are both vast. Hence, the proof rests on demonstrating that the search for the cheapest strategy within a particular sub-family of replicating strategies (an upper bound on the price), and the search for the model which places the highest value on the American option within a particular sub-family of models (a feasible price for the American option that does not create arbitrage opportunities), are the primal and dual of the same finite linear program, and have the same optimal value. Hence the cheapest super-replicating strategy lies in the chosen sub-family, and the model which gives the highest value to the American option is from the given sub-family of models. The methodology provides a viable method of computing the bound in practice.

Duality results of this form are the goal of much of the literature on robust hedging. They can be more or less explicit and/or general/abstract. For specific exotic options (eg lookbacks and barriers) it is sometimes possible to exploit the characteristics of the payoff to describe a model and a super-hedge for which the model-based price and the cost of the super-hedge coincide, thus proving the optimality of both. For general payoffs there are duality results (see, for example, Beiglb\"ock et al~\cite{BeiglbockHenryLaborderePenkner:13}, Bouchard and Nutz~\cite{BouchardNutz:15} and Acciaio et al~\cite{AcciaioBeiglbockPenknerSchachermayer:13})
and our duality result can be seen as an extension of the Super-replication Theorem of Acciaio et al~\cite{AcciaioBeiglbockPenknerSchachermayer:13} from path-dependent claims to American claims, though the technical assumptions are slightly different.

The literature on model free pricing of path-dependent options relies heavily on the duality between pricing and hedging. This duality is used widely in other contexts in mathematical finance, including for the pricing of American options. For example, Andersen and Broadie~\cite{AndersenBroadie:04}, Rogers~\cite{Rogers:02} and Haugh and Kogan~\cite{HaughKogan:04} exploit the relationship between the primal problem of pricing and the dual problem of hedging to bound the value of an American claim. But it should be emphasised that the use of duality in these papers is quite different. They value the American claim within a well-defined model; in their models, the American claim has a precise price. They seek bounds that bracket the true price of the American claim under the given model by using a near to optimal exercise strategy. In the present paper, there is no model, and American claims do not have a unique price. The goal is to find the maximum price of the American claim that does not lead to arbitrage.

There is a closely related body of literature on robust hedging as exemplified by Carr, Ellis and Gupta~\cite{CarrEllisGupta:98} and Carr and Nadtochiy~\cite{CarrNadtochiy:11} that seeks hedging strategies for exotic options that work well across a wide range of models. These strategies do require restrictions on the underlying process, such as the symmetry of the implied volatility surface or the requirement that instantaneous volatility be a deterministic function of the price level, which are shared by a broad range of standard models. Such restrictions are unappealing in the context of the question addressed in this paper where the focus is on the ability of the holder of the American claim to respond to the unexpected.

This paper is organized as follows. Section~\ref{sec:lattice} gives the theory and main results in the case of price processes defined on a bounded rectangular lattice under a simplifying assumption that there is a largest strike at which the call price is zero. The main result is a duality between the pricing and hedging problems. On the pricing side we show that the search over consistent models can be restricted to models in a particular, simple class. This class is wider than the set of one-dimensional Markovian models, (and restricting attention to the Markovian class will only lead to the highest model based price in trivial situations) but is still relatively simple, since it is a class of bivariate Markov processes, with the first dimension as price, and the second dimension the `regime' which switches at the optimal exercise time. On the hedging side we show how the search for the cheapest super-replicating strategy can be restricted to a search over a simple family of super-hedges. A final part of our first theorem shows that there is no duality gap: the highest model based price is equal to the cost of the cheapest super-replicating strategy.

In Section~\ref{sec:extension} we relax some of the lattice assumptions we use in Section~\ref{sec:lattice} in the sense that although we continue to assume that we are given a finite family of European option prices (with strikes and maturities on a grid) we now consider models in continuous time and price process taking values in $\R^+$ rather than a discrete set. In Section~\ref{sec:nocall} we relax the assumption that there is a strike at which call prices are zero. Our final result is again that we can find the supremum over consistent models of the model-based price of the American option, and that this equals the cost of the cheapest super-replicating strategy. In this case the supremum over models may not be attained.

In Section~\ref{ssec:numerical} we argue that the methods of this paper are not purely theoretical, but rather that they provide a viable method of calculating model-independent bounds on the prices of American-style derivatives. We consider an American put, and compare the Black-Scholes value with the model-independent upper bound on the price. Given a set of European option prices we can calculate the
model-independent American option premium. We find that valuation under the Black-Scholes model seriously underestimates the value of the American feature as it fails to take account of the ability of the holder of the American option to change his strategy as uncertainty about the underlying model is resolved.
In Section~\ref{ssec:filtration} we illustrate in a toy example that restricting attention to models in which the filtration is the natural filtration of the price process also significantly underestimates the value of the American feature of the option.
In Section~\ref{ssec:continuum} we show by example how the ideas of the paper can be applied in a continuous setting in which options trade with a continuum of maturities and strikes. Section~\ref{sec:conclusion} concludes.

\section{Processes on a bounded lattice.}
\label{sec:lattice}

\subsection{The set-up}
\label{ssec:setup}

This paper considers the price of an American-style claim on a single underlying stock. Time, denoted by $t$, runs from the current time $t=0$ to some finite positive horizon $T$.

Let $S= (S_{t})$ denote the price of the stock. Let $s_0$ be the initial price of the asset which we view as a known constant. We assume that $S$ is non-negative, and pays no dividends, though the case of an asset which pays proportional dividends can be reduced to this case by considering $S$ as the price of stock after dividends are reinvested. Suppose interest rates are non-stochastic and let $B = (B_t)$ denote the price of a risk free bond (with $B_0=1$). Let $X = (X_{t})$ be given by $X_t = S_{t}/B_{t}$. Then $X$ denotes the price of the asset with bond  numeraire. Finally assume there are no market frictions: there are no transaction costs or taxes and short selling is permitted without restriction.

The American claim is characterized by a function $a_S$ which represents the fact that if the option is exercised at time $t$ then the option holder receives $a_S(S_t,t)$. Let $a(x,t) = a_S(x B_{t}, t)/B_{t}$; then $a$ is the discounted payoff of the American claim, expressed in terms of units of the discounted price $X$. As a motivating example, consider the case of a constant interest rate $r$ and an American put (on $S$) with strike $K_S$. Then we have $B_{t} = e^{rt}$, $a_S(s,t) = (K_S- s)^+$ and $a(x, t) = (K_S e^{-rt} - x)^+$.

In addition to the stock and the pure discount bond, the set of traded securities includes European call options on the stock. In particular, it is possible to buy or sell a call option on $S$ with strike ${K_S}$ and maturity $t$ (ie. with payoff $(S_{t} - {K_S})^+$) for price $C_S({K_S}, t)$ for a finite set of traded strikes and maturities to be described below. Under the bond numeraire this corresponds to being able to buy or sell a call option on $X$ with strike
$K= {K_S}/B_{t}$ and maturity $t$ for a price $C(K,t) = C_S({K}B_{t},t)$.

Henceforth we will work exclusively with the discounted price and with discounted call prices. Moreover, we will omit the qualifier discounted, and instead talk about the prices $X$ and $C$. We expect that under any pricing measure $X=(X_t)$ is a martingale.

In this section we will assume that time is discrete, and that the time parameter is restricted to lie in a set $\sT_0 = \{ t_0 = 0, t_1 < \ldots < t_N = T \}$. Further, we will also assume that the option can only be exercised at a date $\tau \in \sT = \sT_0 \setminus \{0\} = \{ t_1 , \ldots , t_N  \}$.
Later we will extend our analysis to allow the time parameter of the price process and the exercise time of the American option to take values in $\bbT = [0,T]$, although we will still assume that the set of maturities of traded options is finite.

In addition we assume that
for each maturity $t_n \in \sT$ the set of traded strikes is $\sK$ where
\[ \sK = \{ x_1, x_2, \ldots x_{J} \} \]
and $0 < x_1 < x_2 \ldots < x_{J}$.
Since holding a call with strike zero is equivalent to holding the stock, and since the stock is traded, it is useful to consider 0 to be a traded strike. (For any maturity, no-dominance arguments imply that the price of a zero-strike call must equal $X_0=s_0$.)  Let $\sX = \{0, x_1, x_2, \ldots x_{J} \}$. In this section we will identify $\sX$ with a set of levels for the price process $X$ and build processes which live on the lattice $\sX \times \sT$ (at least after time zero). This restriction will be relaxed in future sections.

One rationale for using a finite set of strikes is that this is the `real' situation. However, if strikes for options on $S$ are common across maturity, then after switching to discounted variables this would no longer be the case. Our rationale for a finite set of strikes is primarily pedagogic: this setting provides the simplest situation in which to state and prove the main results, and to illustrate the message of this paper, namely that the value of American options is that they allow agents to take advantage of model uncertainty and to follow strategies which utilise information which is not contained in the natural filtration of the asset.

We assume that the traded European options are calls. We could equally work with puts. In our model an agent can short the asset and so put-call parity holds. Hence it is trivial to switch between working with European puts and calls. Of course, there is no put-call parity for the American option.

\begin{assumption} Time is discrete and takes values in the finite set $\sT_0$. The American option can only be exercised at times $t_n \in \sT$ and must be exercised by $t_N = T$. The price process $X = (X_t)_{t \in \sT_0}$ takes values in $\sX$ for $t>0$.
The payoff function $a:\sX \times \sT \mapsto \R$ is positive. 
\end{assumption}

The assumptions that $a$ is positive and that the option must be exercised are harmless since if not we can simply take the positive part of $a$.

For $0 \leq j \leq J$  and for $1 \leq n \leq N$ write $c_{j,n}$ for the price of a call security paying $(X_{t_n} - x_j)^+$ at time $t_n$. Set $c_{0,n}=s_0$.
Our assumption is that call options can be both bought and sold at time zero for these prices.

\begin{assumption}
\label{ass:conditionsonC}
\begin{enumerate}
\item The set of call option prices has the following properties:
\begin{itemize}
\item For $1 \leq n \leq N$, $s_0 = c_{0,n} \geq c_{1,n} \geq c_{2,n} \geq \cdots \geq c_{J,n} \geq 0$.
\item For $1 \leq n \leq N$,
\( 1 \geq \frac{c_{0,n} - c_{1,n}}{x_1} \geq \frac{c_{1,n} - c_{2,n}}{x_2 -x_1} \geq \cdots \geq \frac{c_{J-1,n} - c_{J,n}}{x_J -x_{J-1}} \) .
\item For $1 \leq n \leq N-1$, and for $0 \leq j \leq J$,
$c_{j, n+1} \geq c_{j,n}$.
\end{itemize}
\item In addition $c_{J,N}=0$.
\end{enumerate}
\end{assumption}

Carr and Madan~\cite{CarrMadan:05} and Davis and Hobson~\cite[Theorem 3.1]{DavisHobson:07} set out necessary and sufficient conditions on a set of call options to ensure the absence of arbitrage. In our setting these conditions reduce to the first set of statements above. The additional hypothesis that $c_{J,N}=0$ (and then also $c_{J,n}=0$) for all $1 \leq n \leq N$ implies that in any model consistent with these option prices, the probability that the option price ever exceeds $x_J$ is zero. This simplifying assumption will be relaxed in Section~\ref{sec:nocall}.

Let ${\bf C}$ be the $(J+1) \times N$ matrix with elements $c_{j,n}$. Define the $(J+1) \times N$ matrix $\bf P$ via its entries $p_{j, n}$ where for $1 \leq n \leq N$
\begin{equation}
\label{eq:pdef}
p_{j,n} = \left\{ \begin{array}{ll}
1 - \frac{s_0 - c_{1,n}}{x_1}  &  j = 0; \\
\frac{c_{j-1,n} - c_{{j},n}}{x_{j} - x_{j-1}} - \frac{c_{j,n} - c_{{j+1},n}}{x_{j+1} - x_j} & 1  \leq j < J; \\
\frac{c_{J-1,n}-c_{J,n}}{x_{J} - x_{J-1}}  & j = J.
\end{array} \right.
\end{equation}

Equation set (\ref{eq:pdef}) is the discrete-space version of the Breeden and Litzenberger~\cite{BreedenLitzenberger:78} formula linking risk neutral densities to the second derivative of option prices with respect to strike. A model in which $\Prob(X_{t_n} = x_j)=p_{j,n}$ for all $j$ and $n$ has the property that $\E[(X_n - x_j)^+] = c_{j,n}$.  There are many other families of marginal distributions which can also deliver these option prices, but this is the only set of probability laws which agree with the call prices for each $n$ if the mass is constrained to lie in the set $\sX$.

\subsection{Consistent pricing models}
\label{ssec:consistentpricingmodels}

\begin{definition}
\label{def:model}
$\sM^{\sX, \sT} = \sM^{\sX, \sT}({\bf C})$ is the set of models (i.e. a filtration $\bbF = (\sF_0, \sF_{t_1}, \ldots \sF_{t_N})$ and a probability measure $\Prob$ supporting a stochastic process $X= (X_{t_n})_{0 \leq n \leq N}$ taking values in $\sX$) such that $X_0 = s_0$, and
\begin{enumerate}
\item the process $X$ is consistent with {\bf C} in the sense that $\E[(X_{t_n}-x_j)^+] = c_{j,n}$ or equivalently $\Prob(X_{t_n} = x_j) = p_{j,n}$;
\item $X$ is a $(\Prob,\bbF)$-martingale.
\end{enumerate}
We say such a model is {\em consistent} with the observed call prices ${\bf C}$.
\end{definition}

The superscript $\sX, \sT$ on $\sM$ refers not to the fact that models are consistent with call prices defined for strikes in $\sX$ and maturities in $\sT$ but rather to the fact that processes are defined on the time parameter set $\sT$, and the price process takes values in $\sX$.

An element $M$ of $\sM^{\sX,\sT}$ not only defines a process which is consistent with ${\bf C}$, it is also defines a pricing model; the model price of a traded security at time $t$ is its conditional expected payoff under $M$. In particular $M$ defines a model based price for the American option: $\phi(M) = \phi^a(M)= \sup_{\tau} \E^M[a(X_\tau, \tau)]$.

\begin{definition}
\label{def:markov}
$\sM_1^{\sX, \sT}({\bf C})$ is the subset of $\sM^{\sX, \sT}({\bf C})$ such that
\begin{enumerate}
\item $X$ is Markovian, so that $\Prob(X_{t_{n+1}} = x_k | \sF_{t_n}) = \Prob(X_{t_{n+1}} = x_k | X_{t_n})$.
\end{enumerate}
We say such a model is a {\em consistent, Markov} model.
\end{definition}

\begin{proposition}
\label{prop:markov}
Suppose that {\bf C} satisfies Assumption~\ref{ass:conditionsonC}. Then $\sM_1^{\sX, \sT}({\bf C})$ is non-empty.
Further, the market comprising the stock, the bond and the call options (trading at the prices {\bf C}) is arbitrage free.
\end{proposition}

\begin{proof} For $1 \leq n \leq N$ let $\mu_n$ denote the atomic measure with mass $p_{j,n}$ at $x_j$, and let $\mu_0$ be the point mass at $s_0$.

The conditions on ${\bf C}$ ensure that the call prices are convex in $x$ (for fixed $n$) and increasing in $n$ for fixed $x$. These are sufficient conditions for there to existence a martingale transport
of $\mu_n$ into $\mu_{n+1}$. This martingale transport can be chosen such that the probability mass transported from $x_j$ to $x_k$ depends on $\mu_n$ and $\mu_{n+1}$ alone. Hence there is a discrete-time martingale (with respect to its natural filtration) which is consistent with the prices in $\bf C$, and which exhibits the Markov property.

The absence of arbitrage follows from the existence of a martingale under which the prices of contingent claims are equal to the expected values of their payoffs (Harrison and Kreps~\cite{HarrisonKreps:79}).
\end{proof}

Note that in the definition of $\sM^{\sX,\sT}$ we do not assume that $\bbF$ is the natural filtration of $X$. The filtration may be considerably richer than this, and the probability space may support other stochastic processes in addition to $X$. We will want $\bbF$ to support (at the least) a second stochastic process, denoted $\Delta$.

In Section~\ref{sec:extension} we will extend the problem to allow for discrete-time price processes taking values in $\R^+$ (so the space of models is $\sM^{\R^+,\sT}({\bf C})$ --- note that although we allow the price process to take any non-negative value, we still assume a finite set $\sX$ of strikes) and then to continuous time processes.
In the case of continuous time we insist that the price process is a right-continuous martingale.

\subsection{Semi-static hedging strategies}
\label{ssec:SSHS}
Now we want to discuss the hedging aspects of the problem. Our set-up includes the notion that we are given option prices for a finite set of vanilla European calls. Then, in addition to allowing investment in the stock, we want also to allow investment in the call options. However, whilst the prices of calls are known today, we do not want to make assumptions about how they will evolve over time (except that they will respect the no arbitrage restrictions in the first part of Assumption~\ref{ass:conditionsonC}). Hence, although we expect to be able to take buy and hold positions in the traded calls, we cannot expect to be able to adjust these portfolios over time --- there is no way to determine what such an adjustment might cost.

\begin{definition}
\label{def:tradingstrategy}
A (path and exercise dependent) semi-static trading strategy $({\bf B}, \Theta= (\Theta^1, \Theta^2))$ on $(\sX, \sT)$ is a composition of
\begin{enumerate}
\item Arrow-Debreu style European options with payoff $(b_{j,n})$ if $X$ is in state $x_j$ at time $t_n$ (for $1 \leq n \leq N$). As securities mature they are held in the bond. The payoff of such a strategy is
    \[ \sG_T^{\bf B} = \sum_{1 \leq n \leq N} \sum_{0 \leq j \leq J} b_{j,n} I_{ \{ X_{t_n} = x_j \} } \]
    and the cost is $\sum_{1 \leq n \leq N} \sum_{0 \leq j \leq J} b_{j,n} p_{j,n}$.
\item A dynamic hedging position of $\Theta_{t_n}$ units of stock created at time $t_n$ for $1 \leq n \leq N-1$. Here $\Theta_{t_n}= \Theta^1(x_{t_1}, \ldots x_{t_n})$ if the option {\em has not} yet been exercised and $\Theta_{t_n} = \Theta^2(x_{t_1}, \ldots x_{t_n}, t_j)$ if the option {\em was} exercised at $t_j$ with $j \leq n$. The position is financed by borrowing and is liquidated at $t_{n+1}$. If exercise occurs at $\rho \in \sT$ then
    the payoff of such a strategy along a price path $(s_0=x_0, x_{t_1}, \ldots x_{t_N})$ is
    \[ \sG_T^{\Theta} = \sum_{n=1}^{\sN(\rho)-1} \Theta^1_{t_n}(x_{t_1},\ldots, x_{t_n})(x_{t_{n+1}} - x_{t_n}) + \sum_{n=\sN(\rho)}^{N-1} \Theta^2_{t_n}(x_{t_1},\ldots, x_{t_n}, \rho)(x_{t_{n+1}} - x_{t_n}), \]
    where $\sN(\rho) = \min \{ n : t_n \geq \rho \}$. The cost is zero.
\end{enumerate}
The time-$T$ payoff $\sG_T = \sG_T^{{\bf B},{\Theta}}$ from the semi-static trading position along a price path $(x_{t_1}, \ldots x_{t_N})$ is
\[ \sG_T(x_{t_1}, \ldots x_{t_N}, \rho)  =  \sG_T^{\bf B} + \sG_T^{\Theta}, \]
and the total cost is $H_{{\bf C}}({\bf B}, \Theta) = H({\bf B}) = \sum_{1 \leq n \leq N} \sum_{0 \leq j \leq J} b_{j,n} p_{j,n}$.
\end{definition}

One might expect to also need to specify a position in the underlying over the time-period $(0,t_1]$, ie to include $\theta^1_0(s_0)$. This is not necessary, since any required payoffs can be subsumed into the European payoffs $b_{j,1}$.

It is normal in the model-independent pricing literature to consider semi-static strategies in which the dynamic element is such that $\Theta_{t_n} = \Theta_{t_n}(x_1, \ldots, x_{t_n})$, ie. such that the position in the stock is a function of the price path to date. In our American option pricing context it is essential that we allow the hedge ratio to also depend on whether the option has been exercised (and then natural to want it to depend on when it was exercised). Hence we need to allow $\Theta_{t_n} = \Theta_{t_n}(x_1, \ldots, x_{t_n}, \rho)$ for $t_n \geq \rho$ where $\rho \in \sT$ is the exercise time.

In principle, in a given model the space of semi-static hedging strategies could depend on the model and could be much richer.
But, the hedger of the option needs to be able to define the gains irrespective of the model. Then, he is constrained to use semi-static strategies where the dynamic component is a function of the price history and the exercise time only, as in Definition~\ref{def:tradingstrategy}.

\begin{definition}
A semi-static trading strategy $({\bf B}, \Theta= (\Theta^1, \Theta^2))$ super-replicates the American claim if
$\sG_T(x_{t_1}, \ldots x_{t_N}, \rho) \geq a(x_\sigma, \rho)$ for all $(x_{t_1}, \ldots x_{t_N})$ with $x_{t_n} \in \sX$ and all $\rho$. Let $\sS = \sS^{\sX,\sT}(a)$ be the set of super-replicating semi-static strategies.
\end{definition}

The superscripts on $\sS$ refer to the fact that the exercise time is in $\sT$ and super-replication occurs along paths for which $x_{t} \in \sX$ for $t \in \sT$.

Define the highest model-based price among models consistent with the prices of the traded calls: $\sP^{\sX,\sT}(a, {\bf C}) = \sup_{M \in \sM^{\sX,\sT}({\bf C})} \phi^a(M)$. Define also the cost of the cheapest super-replicating semi-static strategy $\sH^{\sX,\sT}(a, {\bf C}) =
\inf_{({\bf B, \Theta}) \in \sS^{\sX, \sT}(a)} H_{{\bf C}}({\bf B, \Theta})$.

\begin{proposition}
\label{prop:wd}
Weak duality holds: $\sP^{\sX,\sT}(a, {\bf C}) \leq \sH^{\sX,\sT}(a, {\bf C})$.
\end{proposition}

\begin{proof}
For any semi-static super-hedging strategy $a(X_\tau,\tau) \leq \sG_T(x_{t_1}, \ldots x_{t_N}, \tau)$. Since $X$ is a martingale under any consistent model, if $\tau$ is a stopping time then $\E^M[G^\Theta_T] = 0$ and
\[ \E^M[a(X_\tau,\tau)] \leq \E^M \left[ \sG^{{\bf B}, \Theta}_T \right] = \sum_{1 \leq n \leq N} \sum_{0 \leq j \leq J} b_{j,n}p_{j,n} = H_{{\bf C}}({\bf B}, \Theta) = H_{{\bf C}}({\bf B}). \]
Weak duality follows.
\end{proof}

\subsection{Bounds on the price of the American Option}
\label{ssec:bounds}
Recall that our current setting is discrete-time price processes taking values in $\sX$.

There are many models consistent with the market prices of the European calls. One could, in principle, search among all possible models to find the supremum on the price of the American claim. But the dimensionality of the space of models is vast. As can be seen from Example~\ref{eg:1.1}, one cannot restrict the search to the set of models that are based on the natural filtration.

In this paper, the search is confined to a small subset of models and is formulated as a finite dimensional linear program. The linear program has a dual. It turns out that, for the specific subset of models chosen, the dual program can be interpreted as the search for the cheapest super-replicating strategy in a restricted class of super-replicating strategies. Hence there is a price which is both the cost of a super-replicating strategy and the model price of the American claim for a particular model.
This pair of optimal model in a certain class of models, and optimal super-hedge from a certain class of super-hedges must therefore be the optimal model over all consistent models, and the cheapest super-hedge over all super-replicating strategies.

The representation of the pricing problem and the replication problem as duals is familiar. But there are several points which are worth highlighting. First,
we write the pricing problem as the primal and the replication problem as the dual. This is because it is easier to motivate the choice of the family of models than the family of hedging strategies. Second, it is only because the subset of models is carefully chosen that its dual can be interpreted as the search for the cheapest super-replicating strategy. (Had the subset not included the global supremum, this could not have been the case.) Third, it is not sufficient to consider Markov models for the stock, instead we must consider an augmented process consisting of price and {\em regime}. Fourth, in the dual problem we do not need dynamic hedging strategies which depend on the whole price history, but rather the position in the stock can be made a function of the current price alone, and whether or not the option has been exercised. This is a considerable simplification (and relies on the fact that the American option payoff depends on the current stock price, and not the path history).

The choice of subset of models is critical. Its members must be characterized by a finite --– and reasonably small --- set of parameters to make the search problem tractable. The models must be able to incorporate the initial market values of all the traded securities. Finally, the models need to have the features that make an American claim particularly valuable. The first two considerations suggest that we consider discrete space, Markov jump processes. But the third consideration, taken with Example~\ref{eg:1.1}, suggests that this will not be adequate. American claims become more valuable if the holder can expect to get more information about the distribution of future returns.

Consider, therefore, the following extension $\sM^{\sX,\sT}_2$ of $\sM^{\sX,\sT}_1$.
\begin{definition}
$\sM_2^{\sX, \sT} = \sM_2^{\sX, \sT}({\bf C}) \subseteq \sM^{\sX, \sT}({\bf C})$ is the set of models (i.e. a filtration $\bbF = (\sF_0, \sF_{t_1}, \ldots \sF_{t_N})$ a probability measure $\Prob$ supporting a bivariate, discrete-time, stochastic process $(X,\Delta)= (X_{t_n}, \Delta_{t_n})_{0 \leq n \leq N}$ taking values in $\sX \times \{1,2\}$ for $n \geq 1$) such that $(X_0,\Delta_0) = (s_0,1)$ and
\begin{enumerate}
\item $(X,\Delta)$ is Markov with respect to price, so that $\Prob(X_{t_{n+1}} = x_k | \sF_{t_n}) = \Prob(X_{t_{n+1}} = x_k | X_{t_n}, \Delta_{t_n})$.
\item $\Delta$ is non-decreasing, with $\Delta_{t_N}=2$.
\item the probability that $\Delta_{t_{n+1}}=2$, conditional on $\Delta_{t_n}=1$ depends on $n$ and $X_{t_{n+1}}$ only.
\end{enumerate}
\end{definition}

The last element of this assumption, namely that the transition probabilities of $\Delta$ depend on $X_{t_{n+1}}$ deserves comment. (More normally, in a Markov setting we would expect these probabilities to depend on $X_{t_n}$.)
Typically, we need a certain proportion $q_{k,m}$ of those paths which arrive at $X_{t_{m}}=x_k$ at $t_{m}$ and have $\Delta_{t_{m-1}}= 1$ to have $\Delta_{t_m}=2$.
We could let this proportion depend on the origin of these paths (ie $X_{t_{m-1}}$) but the simplest solution is to make the assumption above. See Remark~\ref{rem:reconstruct} below.

We refer to $\Delta$ as the regime process. The relationship between $\Delta$ and the optimal stopping rule $\tau^*$ is that it will turn out to be optimal
to take $\tau^* = \min \{ t \in \sT : \Delta_t = 2 \}$.

A process $(X,\Delta)$ in $\sM_2^{\sX, \sT}$ can be characterized by a pair of $(J+1) \times (J+1) \times (N-1)$ matrices
${\bf G^1}$ and ${\bf G^2}$ (with entries $g^\delta_{j,k,n}$) specifying the joint probability (and not the conditional probability) of successive states:
\[ g^\delta_{j,k,n} = \Prob(X_{t_n} = x_j, X_{t_{n+1}}=x_k, \Delta_{t_n} = \delta)  \hspace{3mm} 0 \leq j,k \leq J; 1 \leq n \leq N-1, \delta  \in \{ 1,2 \} \]
One might expect to want to specify $g^{\delta}_{j,k,0}$ also, but for $\delta = 2$ these probabilities are necessarily zero, and for $\delta=1$,
$g^{1}_{j,k,0} = p_{k,1}$ if $x_j = s_0$ and zero otherwise. Since these probabilities do not depend on the model (assuming the model is consistent with call prices), we view them as fixed. Indeed, there is no requirement that $s_0 \in \sK$, so it may not be possible to define $g^1$ at $n=0$.

By definition probabilities are positive. Further, the mass entering a node must equal the mass at the node must equal the mass leaving the node. Thus
\begin{equation}
\label{eq:probs} \sum_{0 \leq i \leq J} (g^1_{i,j,n-1} + g^2_{i,j,n-1}) = p_{j,n} = \sum_{0 \leq k \leq J} (g^1_{j,k,n} + g^2_{j,k,n}) \end{equation}
where the equality on the left is defined for $2 \leq n \leq N$ and the equality on the right for $1 \leq n \leq N-1$.

By hypothesis the process $\Delta$ is non-decreasing. It is convenient to introduce an auxiliary $(J+1) \times N$ matrix $\bf F$  which records the probability of arriving at node $(j,2)$ at time $n$ having been in regime $1$ at time $n-1$. Let
${\bf F} = (f_{j,n})$ where $f_{j,n} \geq 0$ is given by the joint probability $f_{j,n} = \Prob(X_{t_n}=j, \Delta_{t_{n-1}}=1, \Delta_{t_n}=2)$. Then
\begin{equation}
\label{eq:fg} f_{j,n}  = \left\{ \begin{array}{ll} \sum_{0 \leq k \leq J} g^2_{j,k,1}  & n=1 \\
                                            \sum_{0 \leq k \leq J} g^2_{j,k,n} - \sum_{0 \leq i \leq J} g^2_{i,j,n-1} \\
                                            \hspace{10mm} = \sum_{0 \leq i \leq J} g^1_{i,j,n-1} - \sum_{0 \leq k \leq J} g^1_{j,k,n}  & 1 < n < N \\
                                            p_{j,N}- \sum_{0 \leq i \leq J} g^2_{i,j,N-1} = \sum_{0 \leq i \leq J} g^1_{i,j,N-1} & n=N \end{array} \right. .
                                            \end{equation}

\begin{remark}
\label{rem:reconstruct}
Given the transition probabilities of $(X,\Delta)$ it is clear that we can calculate ${\bf G^1}$, ${\bf G^2}$ and ${\bf F}$. Conversely, given ${\bf G^1}$, ${\bf G^2}$ and ${\bf F}$ we have $\Prob(X_n = j, X_{n+1}=k) = g^1_{j,k,n} + g^2_{j,k,n}$, so the transitions of $X$ are specified. Moreover, $\Prob(X_{t_n} = j, X_{t_{n+1}}=k, \Delta_{t_n}=1) = g^1_{j,k,n}$. Then
\begin{eqnarray*}
\lefteqn{\Prob(X_{t_n} = j,\Delta_{t_n}=1, X_{t_{n+1}}=k, \Delta_{t_{n+1}}=2) } \\
 &  = &  \Prob(\Delta_{t_{n+1}}=2|X_{t_n} = j, X_{t_{n+1}}=k, \Delta_{t_n}=1) g^1_{j,k,n} \\
& =& \Prob(\Delta_{t_{n+1}}=2| X_{t_{n+1}}=k, \Delta_{t_n}=1) g^1_{j,k,n}
  \; \; = \; q_{k,n+1} g^1_{j,k,n} \end{eqnarray*}
where $q_{k,n+1}$ is chosen so that
\[ q_{j,n+1}\sum_{0 \leq i \leq J} g^1_{i,j,n} = \sum_{0 \leq i \leq J} \Prob(X_{t_n} = i,\Delta_{t_n}=1, X_{t_{n+1}}=j, \Delta_{t_{n+1}}=2) = f_{j,n+1}. \]
Since $f_{j,n+1} = \sum_{0 \leq i \leq J} g^1_{i,j,n} - \sum_{0 \leq k \leq J} g^1_{j,k,n+1}$ we have $q_{j,n+1} \in [0,1]$.
In particular, the matrices ${\bf G^1}$, ${\bf G^2}$ (with or without ${\bf F}$) uniquely determine the transition probabilities of $(X,\Delta)$.
\end{remark}

A further requirement is that process $X$ is a martingale. This implies that for $0 \leq j \leq J$, $1 \leq n \leq N-1$ and $\delta\in \{1,2 \}$
\begin{equation}
\label{eq:mg} \sum_{0 \leq k \leq J} (x_k - x_j) g^\delta_{j,k,n} = 0.  \end{equation}

For any model $M$ which is consistent with the observed call prices we can define the model based price of the American option
by $\phi^a(M) = \sup_\tau \E^M[a(X_\tau,\tau)]$ where the supremum is taken over stopping times $\tau$ and the superscript of the expectation operator refers to the fact that we are taking expectations under the model $M$. Except in Section~\ref{ssec:unrestricted} we will generally suppress the superscript $a$ on $\phi$.
Our goal is to find ${\sP}^{\sX, \sT}(a, {\bf C}) =\sup \phi(M)$, where the supremum is taken over $M \in \sM^{\sX,\sT}({\bf C})$, the space of all discrete-time models in which the price process is a martingale which takes values in $\sX$ and is consistent with call option prices. One of the fundamental contributions of this paper is that is to show that the supremum over $\sM^{\sX,\sT}({\bf C})$ is equal to the supremum over models in the much smaller set $\sM_2^{\sX,\sT}({\bf C})$. Further, given $M \in \sM^{\sX, \sT}_2({\bf C})$ we can define $\tau^* = \inf \{ t \in \sT:\Delta_t = 2 \}$, and $\phi_*(M) = \E^M[a(X_{\tau^*}, \tau^*)]$. Then,
\[ \sup_{M \in \sM^{\sX, \sT}_2({\bf C})} \phi_*(M) \leq \sup_{M \in \sM^{\sX, \sT}_2({\bf C})} \phi(M)  \leq \sup_{M \in \sM^{\sX, \sT}({\bf C})} \phi(M) = {\sP}^{\sX,\sT}(a, {\bf C}). \]
We show there is equality throughout.

Our first task is to find $\sup_{M \in \sM^{\sX,\sT}_2} \phi_*(M)$. Using the conditions (\ref{eq:probs}) and (\ref{eq:mg}) together with (\ref{eq:fg}) this problem can be cast as a linear program. We call this the pricing (primal) problem.

\begin{LP}
\label{LP:primal}
The pricing problem ${\bf L^{\sX,\sT}_P}$ is to: \\
find the $(J+1) \times N$ matrix $F$ and the two $(J+1)\times(J+1)\times(N-1)$ matrices ${\bf G^1}$
and ${\bf G^2}$ which maximise
\[ \sum_{1 \leq n \leq N} \sum_{0 \leq j \leq J} a(x_j, t_n) f_{j,n} \]
subject to ${\bf F} \geq 0$, ${\bf G^1}\geq 0$, ${\bf G^2}\geq 0$, and
\begin{enumerate}
\item[(a)] $\sum_{0 \leq k \leq J} (g^1_{j,k,n} + g^2_{j,k,n}) = p_{j,n}$; \hspace{10mm} $0 \leq j \leq J$, $1 \leq n \leq N-1$.
\item[(b)] $\sum_{0 \leq i \leq J} (g^1_{i,j,n-1} + g^2_{i,j,n-1}) = p_{j,n}$; \hspace{10mm} $0 \leq j \leq J$, $2 \leq n \leq N$.
\item[(c)] $\sum_{0 \leq k \leq J} (x_k - x_j) g^1_{j,k,n} = 0$; \hspace{10mm} $0 \leq j \leq J, 1 \leq n \leq N-1$.
\item[(d)] $\sum_{0 \leq k \leq J} (x_k - x_j) g^2_{j,k,n} = 0$; \hspace{10mm} $0 \leq j \leq J, 1 \leq n \leq N-1$.
\item[(e)] \[ \left\{ \begin{array}{ll} f_{j,1} - \sum_{0 \leq k \leq J} g^2_{j,k,1} \leq 0 & \\
                                f_{j,n} - \sum_{0 \leq k \leq J} g^2_{j,k,n} + \sum_{0 \leq i \leq J} g^2_{i,j,n-1} \leq 0 & 1<n<N \\
                                f_{j,N} + \sum_{0 \leq k \leq J} g^2_{i,j,N-1} \leq p_{j,N} &
                  \end{array} \right.  \]
\end{enumerate}
Let the optimum value be given by $\Phi^{\sX,\sT}= \Phi^{\sX,\sT}(a, {\bf C})$.
\end{LP}

\begin{remark}
\label{rem:Jabsorbing}
It follows from (c) and (d) that we must have $g^\delta_{J,k,n}=0$ for $k<J$, so that for any feasible model, $\{x_J\}$ is absorbing.
\end{remark}

The inequalities in (e) are actually equalities, recall (\ref{eq:fg}). However, since the coefficients in the objective function are positive and since we seek to maximise $\phi$ we can write them as inequalites, and we will obtain equality in the optimal solution. Moreover, by writing (e) as a set of inequalities we will end up with fewer constraints in the dual problem. Strict inequality corresponds to
\[ \sum_{1 \leq n \leq N} \sum_{0 \leq j \leq J} f_{j,n} < \sum_{0 \leq j \leq J} p_{j,N} = 1 \]
and a failure to exercise the American option in some scenarios. Clearly this is suboptimal unless $a(x_j,t_n)=0$ for some $j$ and $n$.

\begin{proposition}
\label{prop:primalsolution}
${\bf L^{\sX,\sT}_P}$ is a linear program for which the feasible set is non-empty and the objective function is bounded. There exists an optimal solution.
\end{proposition}

\begin{proof}
The fact that ${\bf L^{\sX,\sT}_P}$ is a linear programme follows by inspection. To show that the feasible set is non-empty we need to construct an element of $\sM^{\sX,N}_2$. But $\sM^{\sX,N}_1$ is non-empty. Let $M^1$ be the associated model with $X$ the associated price process, and let $\Delta$ be the process which switches regime at time 1 so that $\Delta_n=2$ for $n \geq 1$. We have
\begin{eqnarray*}
g^1_{j,k,n}  =  0, && 0 \leq j,k \leq J, 1 \leq n \leq N \\
g^2_{j,k,n}  = \Prob^{M^1}( X_{n} = x_j, X_{n+1}=x_k) && 0 \leq j,k \leq J, 1 \leq n \leq N \\
f_{j,n}     & = & \left\{ \begin{array}{ll} p_{j,1} & 0 \leq j \leq J, n=1 \\
                                            0 & 0 \leq j \leq J, 1 < n \leq N \end{array} \right.
\end{eqnarray*}
and $(X,\Delta) \in \sM^{\sX,\sT}_2$.

Clearly for a general element of $\sM^{\sX,\sT}_2$ we have ${\bf F} \leq {\bf P}$ and hence $\phi_*(M) \leq \sum_{1 \leq n \leq N} \sum_{0 \leq j \leq J} a(x_j, t_n) p_{j,n} < \infty$. The existence of an optimal solution follows.
\end{proof}

\subsection{The hedging problem}
\label{ssec:hedging}

\begin{LP}
\label{LP:dual}
The hedging problem ${\bf L^{\sX,\sT}_H}$ is to: \\
find the three $(J+1) \times N$ matrices ${\bf E^1}$, ${\bf E^2}$ and ${\bf V}$ and the two $(J+1)\times(N-1)$ matrices ${\bf D^1}$
and ${\bf D^2}$ which minimise
\[ \sum_{0 \leq j \leq J, 1 \leq n \leq N} (e^1_{j,n} + e^2_{j,n})p_{j,n} + \sum_{0 \leq j \leq J} v_{j,N} p_{j,N}  \]
subject to ${\bf V} \geq 0$, and
\begin{enumerate}
\item[(i)] for $0 \leq j \leq J$, $1 \leq n \leq N$
\begin{equation}
\label{eq:lpi} v_{j,n} \geq a(x_j,t_n);
\end{equation}
\item[(ii)] for $0 \leq j,k \leq J$, $1 \leq n \leq N-1$
\begin{equation}
\label{eq:lpii}
e^1_{j,n} + e^2_{k,n+1} + (x_k-x_j)d^1_{j,n} \geq 0 ;
\end{equation}
\item[(iii)] for $0 \leq j,k \leq J$, $1 \leq n \leq N-1$,
\begin{equation}
\label{eq:lpiii}
e^1_{j,n} + e^2_{k,n+1} + (x_k-x_j)d^2_{j,n} - v_{j,n} + v_{k,n+1} \geq 0;
\end{equation}
\end{enumerate}
and $e^1_{j,N}=e^2_{j,1}=0$.
Let the optimum value be given by $\Psi^{\sX,\sT} = \Psi^{\sX,\sT}(a, {\bf C})$.
\end{LP}

\begin{proposition}
\label{prop:primaldualequality}
${\bf L^{\sX,\sT}_H}$ is the dual problem to ${\bf L^{\sX,\sT}_P}$. Moreover the optimal solution to ${\bf L^{\sX,\sT}_H}$ exists and the value of ${\bf L^{\sX,\sT}_H}$ is equal to the value of ${\bf L^{\sX,\sT}_P}$.
\end{proposition}

\begin{proof}
Constraints (i), (ii) and (iii) of ${\bf L^{\sX,\sT}_H}$ correspond to the variables ${\bf F}$, ${\bf G^1}$ and ${\bf G^2}$ in ${\bf L^{\sX,\sT}_P}$ respectively, whilst constraints (a) to (e) of ${\bf L^{\sX,\sT}_P}$ correspond to the variables ${\bf E^1}$, ${\bf E^2}$, ${\bf D^1}$, ${\bf D^2}$ and ${\bf V}$.

The two problems are duals and ${\bf L^{\sX,N}_P}$ has an optimal solution (Proposition~\ref{prop:primalsolution}). Hence by the Strong Duality Theorem (Vanderbei~\cite{Vanderbei:08}) an optimal solution to the dual exists and has equal to the value of the primal problem.
\end{proof}



Note that in general we do not expect the dual problem to have a unique optimiser.

Although we called ${\bf L^{\sX,\sT}_H}$ the hedging problem, so far this is purely a statement of nomenclature which needs to be justified.
The next step is to show that the linear program ${\bf L^{\sX,\sT}_H}$ can be interpreted as the search for the cheapest member of a set of super-replicating strategies for the American claim.

\begin{definition}
\label{def:tradingstrategyEEDDV}
Given three $(J+1) \times N$ matrices ${\bf E^1}$, ${\bf E^2}$ and ${\bf V}$ and two $(J+1)\times(N-1)$ matrices ${\bf D^1}$
and ${\bf D^2}$, the quintuple $({\bf E^1}, {\bf E^2}, {\bf D^1}, {\bf D^2}, {\bf V})$ can be interpreted as a semi-static trading strategy for the agent in the following sense:
\begin{enumerate}
\item Let $b_{j,n} = (e^1_{j,n} + e^2_{j,n})$ for $1 \leq n \leq N-1$ and
$b_{j,N} =(e^1_{j,N} + e^2_{j,N} + v_{j,N})$.
\item Let $\theta^1_{t_n}(x_{t_1}, \ldots, x_{t_n}) = \theta^1_{t_n}(x_{t_n}) = d^1_{j,n}$ if $x_{t_n} = x_j$.
\item Let $\theta^2_{t_n}(x_{t_1}, \ldots, x_{t_n}, \sigma) = \theta^2_{t_n}(x_{t_n}) = d^2_{j,n}$ if $x_{t_n} = x_j$.
\end{enumerate}
We call a strategy of this form a {\em Markovian} semi-static strategy.
\end{definition}

\begin{proposition}
If the quintuple $({\bf E^1}, {\bf E^2}, {\bf D^1}, {\bf D^2}, {\bf V})$ is feasible for ${\bf L^{\sX,\sT}_H}$ and if $x_{t_n} \in \sX$ for $1 \leq n \leq N$ then the Markovian semi-static trading strategy in Definition~\ref{def:tradingstrategyEEDDV} super-replicates the American claim.
\end{proposition}
\begin{proof}
For each of $h=\{ e^1, e^2, d^1, d^2, v \}$ write $h_n(x) = \sum_{0 \leq j \leq J} h_{j,n} I_{ \{ x = x_j \} }$.

Suppose that $X$ follows the path $(s_0, y_1, \ldots, y_N)$ with $y_i \in \sX$. The terminal payoff $\sG_T= \sG_T(y_1, \ldots, y_N, \tau)$ to the strategy described in Definitions~\ref{def:tradingstrategy} and \ref{def:tradingstrategyEEDDV} is
\[ \sG_T = \sum_{n=1}^N  (e^{1}_{n}(y_n) + e^2_{n}(y_n))  + v_N(y_N)
+ \sum_{1}^{\sN(\tau)-1} (y_{n+1}- y_n) d^1_{n}(y_n)
+ \sum_{\sN(\tau)}^{N-1}  (y_{n+1}- y_n) d^2_{n}(y_n)
\]
This can be rewritten as
\begin{eqnarray*} \sG_T & = & e^2_1(y_1) + e^1_N(y_N) + \{ v_{\sN(\tau)}(y_\tau) - a(y_\tau,\tau) \} \\
& & + \sum_{1}^{\sN(\tau)-1} \left\{ e^1_n(y_n) + e^2_{n+1}(y_{n+1}) + (y_{n+1}- y_n) d^1_{n}(y_n) \right\} \\
& & + \sum_{\sN(\tau)}^{N-1} \left\{ e^1_n(y_n) + e^2_{n+1}(y_{n+1}) + (y_{n+1}- y_n) d^2_{n}(y_n) - v_n(y_n) + v_{n+1}(y_{n+1}) \right\} \\
& & + a(y_\tau,\tau)
\end{eqnarray*}
The first two elements are zero, and the next three are non-negative due to the feasibility of the quintuple $({\bf E^1}, {\bf E^2}, {\bf D^1}, {\bf D^2}, {\bf V})$. It follows that $\sG_T(y_1, \ldots, y_N) \geq a(y_\tau,\tau)$ and hence for every possible path $(y_1, \ldots y_N)$ in $\sX^N$, and for every possible stopping rule $\tau$ the strategy in Definition~\ref{def:tradingstrategyEEDDV} super-replicates.
\end{proof}

\begin{theorem}
\label{thm:main1} $\Phi^{\sX,\sT} = {\sP}^{\sX,\sT}(a, {\bf C}) = {\sH}^{\sX,\sT}(a, {\bf C}) = \Psi^{\sX,\sT}$. In particular,
under a modelling assumption that for $t \in \sT$ the price process only takes values in $\sX$, the most expensive model-based price amongst models which are consistent with the observed call prices is attained by a price/regime model (an element of $\sM_2^{\sX,\sT}$). Similarly, there is a super-replicating Markovian semi-static strategy for which the cost of the strategy is the lowest amongst the class of all super-replicating semi-static strategies.
\end{theorem}

\begin{proof}
By weak duality, (Proposition~\ref{prop:wd}) the fact that $\sM_2^{\sX, \sT} \subseteq \sM^{\sX,\sT}$ and the fact that the strategy in Definition~\ref{def:tradingstrategyEEDDV} super-replicates we have $\Phi^{\sX,\sT} \leq {\sP}^{\sX,\sT}(a, {\bf C}) \leq {\sH}^{\sX,\sT}(a, {\bf C}) \leq \Psi^{\sX,\sT}$.
But Proposition~\ref{prop:primaldualequality} implies that $\Phi^{\sX,\sT} = \Psi^{\sX,\sT}$ and hence there is equality throughout.
\end{proof}

\subsection{An example}
\label{ssec:example}

The following example is an extension and reformulation of Example~\ref{eg:1.1} to the current setting.

The current price of the underlying is 100. European call options trade with maturities in $\sT = \{ t_1, \ldots ,t_N = T \}$ and strikes in $\sK = \{ 50, 100, 150 \}$. Let $\sX = \{ 0 \} \cup \sK$. Let $(q_m)_{1 \leq m \leq N}$ be a set of probabilities which sum to 1.

Define the set of call option prices by ${\bf C} = c_{j,n}$ where for $1 \leq n \leq N$
\[  c_{j,n} = \left\{ \begin{array}{ll} 100  &  j=0  \\
                                        50  & j = 1 \\
                                        25 \sum_{i = 1}^{n} q_i &  j=2 \\
                                        0  & j=3      \end{array}   \right. \]

The simplest model consistent with option prices is one in which at some time $t \in \sT$ the price jumps from 100 to either 50 or 150. The price levels 50 and 150 are absorbing. The probability that the jump occurs at time $t_n$ for $n \in \{1,2, \ldots, N \}$ is $q_n$. Martingale considerations imply that if there is a jump the probability of an up jump (to 150) is equal to the probability of a down jump (to 50).

Consider now an American option which has payoff $a(x,t_n) = (b_n - x)^+$ where $(b_n)_{ n \in \sN = \{1, \ldots, N \}}$ is a decreasing sequence of numbers with $100< b_1 < 150$. The option must be exercised at one of the dates $\{ t_1, \ldots t_N \}$. Set $a_{j,n} = a(x_j,t_n)$ so that $a_{0,n} = b_n$, $a_{1,n} = (b_n - 50)^+$, $a_{2,n} = (b_n -100)^+$ and $a_{3,n}=0$.

Define
\[ n^* = \max_{n \geq 1} \left\{ n : (b_n - 50) > 2(b_1 - 100) \right\} .\]
By the monotonicity of $b_n$ we have $(b_n - 50) > 2(b_1 - 100)$ for all $n \leq n^*$. Since $b_1 < 150$ we must have $n^*\geq 1$. We suppose $b_N \leq 2b_1 - 150$ so that $n^* < N$.

For the primal pricing problem define ${\bf G^1}$ and ${\bf G^2}$ via
\[ g^1_{2,1,n} = \frac{q_{n+1}}{2} I_{ \{ n \leq n^* - 1 \} } \hspace{10mm} g^1_{2,2,n} = \sum_{n+2}^{n^*} q_i  \hspace{10mm} g^1_{2,3,n} = \frac{q_{n+1}}{2} I_{ \{ n \leq n^* - 1 \} } \]
\[ g^2_{1,1,n} = \frac{1}{2} \sum_{m \leq n}q_m \hspace{10mm} g^2_{2,1,n} = \frac{q_{n+1}}{2} I_{ \{ n \geq n^* \} } \hspace{10mm} g^2_{2,2,n} = \sum_{ (n^* +1) \vee (n+2)}^N q_{m} \]
\[  g^2_{2,3,n} = \frac{q_{n+1}}{2} I_{ \{ n \geq  n^* \} } \hspace{5mm} g^2_{3,3,n} = \frac{1}{2} \sum_{m \leq n}q_m \]
with all other entries being zero.
It follows that the entries of ${\bf F}$ are given by
\[ f_{1,n} = \frac{q_n}{2} I_{ \{ n\leq n^* \} } \hspace{10mm} f_{2,n} = \left(\sum_{n^*+1}^N q_i\right) I_{ \{ n=1 \} }  \hspace{10mm} f_{3,n} = \frac{q_n}{2} I_{ \{ n\leq n^* \} } \]
and that ${\bf F}$, ${\bf G^1}$ and ${\bf G^2}$ satisfy the feasibility conditions of Linear Program~\ref{LP:primal}.
For this set of transition probabilities the model based price of the American call (using the stopping time $\tau =  \inf \{ t_m \in \sT : \Delta_{t_m} = 2 \}$) is
\[ \Phi = \sum_{j,n}f_{j,n}a_{j,n} = (b_1 - 100) \sum_{n^*+1}^N q_i   +  \sum_1^{n^*} \frac{q_n}{2}(b_n - 50) \]
Note that in this model, we may consider the jump time as known at time 1. If the jump time is at or before $t_{n^*}$ exercise is delayed until the time of the jump; if the jump time is at or after $t_{n^*+1}$ then it is not optimal to wait, but instead the American option should be exercised immediately, at time 1.

Now consider the dual hedging problem. Set ${\bf D^1}=0$, ${\bf E^2}=0$ and define ${\bf V}$, ${\bf D^2}$ and ${\bf E^1}$ by
\[ \begin{array}{rcl}
v_{0,n} & = & {\cred \max \{ b_n, 3(b_1 - 100) \} } \\
v_{1,n} & = & (b_n - 50) I_{ \{ n \leq n^* \} } + 2(b_1 - 100) I_{ \{ n > n^* \} } \\
v_{2,n} & = & (b_1 - 100) \\
v_{3,n} & = & 0
\end{array}
 \]
(see Figure~\ref{fig:eg}) together with, for $1 \leq n < N$, $e^1_{j,n} = (v_{j,n} - v_{j,n+1})$ and for $0 \leq j < 3$, $d^2_{j,n} = (v_{j+1,n+1} - v_{j,n+1})/50$ with $d^2_{J,n}=0$.

Since ${\bf E^1} \geq 0$ it follows that (\ref{eq:rtp1}) holds. For (\ref{eq:rtp})
note that $e^1_{j,n} + (x_k - x_j) d^2_{j,n} + v_{k,n+1} - v_{j,n} = ( e^1_{j,n} + v_{j,n+1} - v_{j,n} ) + ( v_{k,n+1} - v_{j,n+1} - (x_k - x_j) d^2_{j,n}) \geq 0$
where we use the fact that $v_{j,n} =  e^1_{j,n} + v_{j,n+1}$ and $\bar{v}_{n+1}$ is convex, so that $v_{k,n+1} \geq v_{j,n+1} + (x_k - x_j) d^2_{j,n}$ as long as $d^2_{j,n}$ is in the subdifferential of $\bar{v}_{n+1}$. Then the feasibility conditions of the dual problem are satisfied.

Further, $\Psi = \sum_{j,n} (e^1_{j,n} + e^2_{j,n}) p_{j,n} + \sum_j v_{j,N} p_{j,N}$ is given by
\begin{eqnarray*}
\Psi & = & \sum_1^{n^*-1} (b_n - b_{n+1}) \sum_{m \leq n } \frac{q_m}{2} + [(b_{n^*}-50) - 2(b_1 - 100)] \sum_{m \leq n^* } \frac{q_m}{2}  + (b_1 - 100) \\
& = & \sum_1^{n^*-1} b_n  \sum_{m \leq n } \frac{q_m}{2} - \sum_2^{n^*} b_n \sum_{m \leq n-1 } \frac{q_m}{2} + (b_{n^*}-50) \sum_{m \leq n^* } \frac{q_m}{2} + (b_1 - 100) \sum_{n^*+1}^N q_m\\
& = &   \sum_1^{n^*} (b_n  - 50) \frac{q_n}{2} + (b_1 - 100)\sum_{n^*+1}^N q_m
\end{eqnarray*}
Hence the candidate solutions for the primal and dual problems yield the same value for the corresponding linear programme, and must both be optimal.

\begin{figure}[!htbp]

\centering

\begin{tikzpicture}[scale=0.75] 
 \draw[-] (0,0) -- (0,8) ;
 \draw[-] (0,0) -- (10,0) ;

 {\cgreen
 \draw[-] (0,6.8) -- (3,4) ;
 \draw[-] (3,4) -- (6,2) ;
 \draw[-] (6,2) -- (9,0) ;
 \draw[dotted] (0,6.8) -- (6.8,0) ;
 }
 {\cblue
 \draw[-] (0,7.4) -- (3,4.4) ;
 \draw[-] (3,4.4) -- (6,2) ;
 \draw[dotted] (3,4.4) -- (7.4,0) ;
}
 \draw[red] (0,7.8) -- (3,4.8) ;
 \draw[red] (3,4.8) -- (6,2) ;
 \draw[red,dotted] (3,4.8) -- (7.8,0) ;

 \draw[dashed] (0,2) -- (6,2) ;
 \draw[dashed] (0,4.8) -- (3,4.8) ;
 \draw[dashed] (0,4.4) -- (3,4.4) ;
 \draw[dashed] (0,4) -- (3,4) ;
 \draw[dashed] (3,0) -- (3,4.8) ;
 \draw[dashed] (6,0) -- (6,2) ;

\draw [] (3,0) circle [radius=.0] node [below]{50} ;
\draw [] (6,0) circle [radius=.0] node [below]{100} ;
\draw [] (9,0) circle [radius=.0] node [below]{150} ;
\draw [] (0,2) circle [radius=.0] node [left]{$(b_1-100)$} ;
\draw [green] (0,3.8) circle [radius=.0] node [left]{$2(b_1-100)$} ;
\draw [blue] (0,4.4) circle [radius=.0] node [left]{$b_{n^*}-50$} ;
\draw [red] (0,5.0) circle [radius=.0] node [left]{$b_{n^*-1}-50$} ;
\draw [green] (0,6.8) circle [radius=.0] node [left]{$b_{n^*+1}$} ;
\draw [blue] (0,7.4) circle [radius=.0] node [left]{$b_{n^*}$} ;
\draw [red] (0,7.9) circle [radius=.0] node [left]{$b_{n^*-1}$} ;
\draw [blue] (7.4,0) circle [radius=.0] node [below]{$b_{n^*}$} ;
\end{tikzpicture}%

\caption{A plot of the function $v$ as a function of strike and maturity, linearly interpolated across strikes. As maturity increases the colour changes from red (maturity $n^*-1$) to blue ($n^*$) to green ($n^*+1$). Also show by the dotted lines are the payoff $a$ of the American option on immediate exercise.  }

\label{fig:eg}
\end{figure}
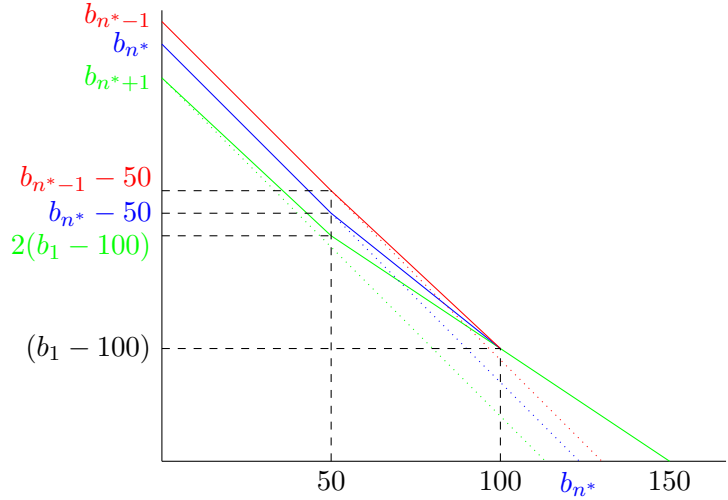

\section{Extensions to processes on $\R^+ \times \bbT$}
\label{sec:extension}
Our goal in this section is to show that the assumptions of the previous section that the price process is restricted to take values in $\sX$ and that the exercise time is restricted to takes values in $\sT$ are not important to the general result, and that similar pricing and hedging results hold true in a more general framework under some mild extra assumptions on the payoff of the American claim. First, we show that over a much wider class of discrete-time models which are consistent with the observed call prices on $\sX \times \sT$ but take values in $\R_+$, the highest model based price is still given by $\Phi^{\sX,\sT}$, the cheapest super-hedge is still given by $\Psi^{\sX,\sT}$, and $\Phi^{\sX,\sT} = \Psi^{\sX,\sT}$ as before. Second, we show that we can extend the results to allow for exercise at arbitrary times $t \in \bbT$, and not just times in $t \in \sT$.

At this stage the key assumption that remains in force is that $c_{J,N}=0$.

\subsection{Processes on $[0,x_J] \times \sT$}
\label{ssec:interval}

\begin{assumption} Time is discrete and takes values in the finite set $\sT_0$. The price process $X = (X_t)_{t \in \sT_0}$ takes values in $[0,x_J]$. $a$ is defined on $[0,x_J] \times \sT$ and that in addition to being positive, $a$ is convex in its first argument. \end{assumption}

Given a function $h$ defined on $\sX$ we can define the {\em linear interpolation} $\bar{h}$ on $[0,x_J]$ of $h$ via
\[ \bar{h}(x) = \frac{x_{j+1} - x}{x_{j+1}-x_j} h(x_j) + \frac{x - x_j}{x_{j+1}-x_j} h(x_{j+1})  \hspace{10mm} x_j \leq x \leq x_{j+1}; 0 \leq j < J. \]

We will need a second type of interpolation for the functions $d^\delta_n$. For $\delta \in \{1,2\}$ we define the {\em mixed interpolation}
$\tilde{d}^{\delta}_n$ by $\tilde{d}^{\delta}_n(x) = {d}^{\delta}_{j,n}$ for $x \in \sX$ and for $x \in (x_j, x_{j+1})$
\begin{equation}
\label{eq:tilded}
 \tilde{d}^{\delta}_n(x) = \left\{ \begin{array}{lcl} d^\delta_{j,n} &\; & d^\delta_{j,n} \leq u^\delta_{j,n} \\
                                                     d^\delta_{j+1,n} & &d^\delta_{j,n} > u^\delta_{j,n} \mbox{ and } d^\delta_{j+1,n} \geq  u^\delta_{j,n} \\
                                                     u^\delta_{j,n} & & d^\delta_{j+1,n} < u^\delta_{j,n} < d^\delta_{j,n} \end{array} \right.
                                                     \end{equation}
where $u^1_{j,n} = (e^1_{j+1,n} - e^1_{j,n})/(x_{j+1}-x_j)$ and $u^2_{j,n} = [(e^1_{j+1,n} - v_{j+1,n}) - (e^1_{j,n}-v_{j,n})]/(x_{j+1}-x_j)$.
Note that for all $x \in [0,x_J]$,
\begin{equation}
\label{eq:dgeqmind}
\tilde{d}^{\delta}_n(x) \geq \min_{0 \leq j \leq J} d^\delta_{j,n}
\end{equation}

\begin{proposition}
\label{prop:linearized}
Suppose the quintuple $({\bf E^1}, {\bf E^2}, {\bf D^1}, {\bf D^2}, {\bf V})$ satisfy the feasibility conditions of the hedging problem in Linear Program~\ref{LP:dual}. Then if we take the linear interpolations (in space) $({\bf \bar{E}^1}, {\bf \bar{E}^2},{\bf \bar{V}})$ of $({\bf E^1}, {\bf E^2},{\bf V})$ and the mixed interpolations $({\bf \tilde{D}^1}, {\bf \tilde{D}^2})$ of $({\bf D^1}, {\bf D^2})$ then the quintuple
$({\bf \bar{E}^1}, {\bf \bar{E}^2}, {\bf \tilde{D}^1}, {\bf \tilde{D}^2}, {\bf \bar{V}})$ satisfy
\begin{eqnarray}
\label{eq:rtp1}
\bar{e}^1_{n}(x) + \bar{e}^2_{n+1}(y) + (y - x) \tilde{d}^1_{n}(x) & \geq & 0  \\
\label{eq:rtp}
\bar{e}^1_{n}(x) + \bar{e}^2_{n+1}(y) + (y - x) \tilde{d}^2_{n}(x) - \bar{v}_{n}(x) + \bar{v}_{n+1}(y) & \geq & 0
\end{eqnarray}
for all $0 \leq x,y \leq x_J$.
\end{proposition}

\begin{proof} We prove (\ref{eq:rtp}), (\ref{eq:rtp1}) being similar, but easier.
We suppose that $e^1_{j,n} + e^2_{k,n+1} + (x_k - x_j) d^2_{j, n} - v_{j,n} + v_{k,n+1} \geq 0$ for all $0 \leq j,k \leq J$ and for fixed $n$ with $1 \leq n \leq N-1$ and aim to deduce that (\ref{eq:rtp}) holds for all $0 \leq x,y \leq x_J$.

Define $h(x_j,x_k) = e^2_{j,n} + e^2_{k,n+1} + (x_k - x_j) d^2_{j, n} - v_{j,n} + v_{k,n+1}$ and $\bar{h}(x, x_k) = \bar{e}^2_{n}(x) + e^2_{k,n+1} + (x_k - x) \tilde{d}^2_{n}(x) - \bar{v}_{n}(x) + v_{k,n+1}$.

Suppose first that $y \in \sX$. If $x \in \sX$ then (\ref{eq:rtp}) follows automatically. So suppose $x \notin \sX$ and write $x = \alpha x_j + (1-\alpha) x_{j+1}$ for some $j$ and $0 < \alpha < 1$.
Then, if $y = x_i$ and $d^2_{j,n} \leq u^2_{j,n} = \frac{(e^1_{j+1,n} - v_{j+1,n})- (e^1_{j,n}-v_{j,n})}{x_{j+1}- x_j}$ we have $\tilde{d}^2_n(x_j) = d^2_{j,n}$ and
\begin{eqnarray*}
\bar{h}(x, y) - \bar{h}(x_j,y) & = & (1-\alpha)[e^1_{j+1,n} - e^1_{j,n} - v_{j+1,n} + v_{j,n} - (x_{j+1} - x_j) d^2_{j,n}]  \\
& = & (1-\alpha) (x_{j+1}- x_j) [ u^2_{j,n} - d^2_{j,n}] \geq 0 .
\end{eqnarray*}
Hence $\bar{h}(x, y) \geq \bar{h}(x_j,y) \geq 0$. Similarly, if $d^2_{j+1,n} \geq u^2_{j,n}$ we find $\bar{h}(x, y) \geq \bar{h}(x_{j+1},y) \geq 0$.

The remaining case is when $d^2_{j+1,n} < u^2_{j,n} < d^2_{j,n}$. Then, if $y \leq x_j$
\[  \bar{h}(x, y) - \bar{h}(x_j,y)  =  (1-\alpha)(y - x_j) (u^2_{j,n} - d^2_{j,n}) \geq 0 \]
whereas, if $y \geq x_{j+1}$,
\[  \bar{h}(x, y) - \bar{h}(x_{j+1},y)  =  (1-\alpha)(y - x_{j+1}) (u^2_{j,n} - d^2_{j+1,n}) \geq 0. \]
For a pictorial representation of these arguments in the case of (\ref{eq:rtp1}) see Figure~\ref{fig:mixedinter}.

It follows that $\bar{h}(x,y) \geq 0$ for all $x \in [0,x_J]$ and for all $y \in \sX$. We want to deduce that
(\ref{eq:rtp}) holds for all $x,y \in [0,x_J]$. But, for fixed $x \in \sX$, the expression on the left-hand-side of (\ref{eq:rtp}) is piecewise linear, with kinks at points $y \in \sX$. Thus, if it is non-negative on $\sX$ it is non-negative for all $y \in [0,x_J]$.

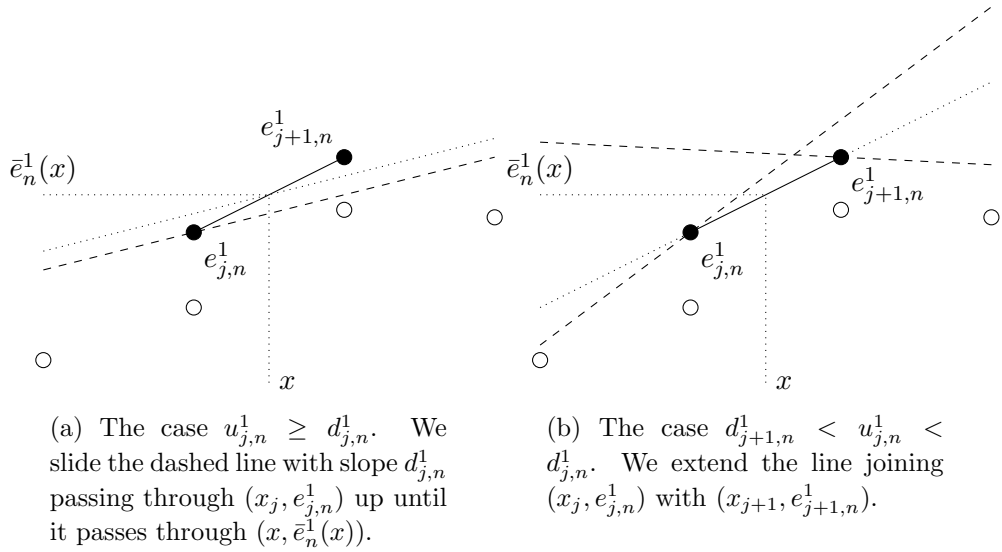
\begin{figure}[!htbp]
\captionsetup[subfigure]{width=0.4\textwidth}
\centering

\subcaptionbox{The case $u^1_{j,n} \geq d^1_{j,n}$. We slide the dashed line with slope $d^1_{j,n}$ passing through $(x_j,e^1_{j,n})$ up until it passes through $(x,\bar{e}^1_n(x))$.}[.5\textwidth]{
\begin{tikzpicture}
 \draw[-] (2,2) -- (4,3) ;
 \draw[dashed] (0,1.5) -- (6,3);
 \draw[dotted] (0,1.75) -- (6,3.25) ;
 \draw[dotted] (3,0) -- (3,2.5) ;
 \draw[] (3,0) circle [radius=0] node [right] {$x$} ;
 \draw[dotted] (0,2.5) -- (3,2.5) ;
 \draw[] (0,2.5) circle [radius=0] node [above] {$\bar{e}^1_n(x)$} ;
 \draw [fill] (2,2) circle [radius=.1] node [below right] {$e^1_{j,n}$};
  \draw [fill] (4,3) circle [radius=.1] node [above left] {$e^1_{j+1,n}$};
  \draw [] (0,0.3) circle [radius=.1] ;
  \draw [] (2,1.0) circle [radius=.1] ;
  \draw [] (4,2.3) circle [radius=.1] ;
  \draw [] (6,2.2) circle [radius=.1] ;
\end{tikzpicture}}%
\subcaptionbox{The case $d^1_{j+1,n} < u^1_{j,n} < d^1_{j,n}$. We extend the line joining $(x_j,e^1_{j,n})$ with $(x_{j+1},{e}^1_{j+1,n})$.}[.5\textwidth]{
\begin{tikzpicture}
\draw[-] (2,2) -- (4,3) ;
 \draw[dashed] (0,0.5) -- (6,5);
 \draw[dashed] (0,3.2) -- (6,2.9) ;
 \draw[dotted] (0,1) -- (2,2) ;
 \draw[dotted] (4,3) -- (6,4) ;
 \draw[dotted] (3,0) -- (3,2.5) ;
 \draw[] (3,0) circle [radius=0] node [right] {$x$} ;
 \draw[dotted] (0,2.5) -- (3,2.5) ;
 \draw[] (0,2.5) circle [radius=0] node [above] {$\bar{e}^1_n(x)$} ;
 \draw [fill] (2,2) circle [radius=.1] node [below right] {$e^1_{j,n}$};
  \draw [fill] (4,3) circle [radius=.1] node [below right] {$e^1_{j+1,n}$};
  \draw [] (0,0.3) circle [radius=.1] ;
  \draw [] (2,1.0) circle [radius=.1] ;
  \draw [] (4,2.3) circle [radius=.1] ;
  \draw [] (6,2.2) circle [radius=.1] ;
\end{tikzpicture}}
\caption{
A pictorial verification of (\ref{eq:rtp1}) for $y \in \sX$. The open circles represent points $(x_k, -e^2_{k,n+1})$.
We want to show there is a line passing through $(x, \bar{e}^1_n(x))$ which lies above $(x_k, -e^2_{k,n+1})$ for all $k$. In the first case, Panel (a), we are given that there exists a dashed line with slope $d^1_{j,n}$ which lies above the circles. Since $(x, \bar{e}^1_n(x))$ lies above the dashed line we can slide the dashed line up until it becomes the sloping dotted line, which passes through $(x, \bar{e}^1_n(x))$ and also lies above the open circles. In the second case, Panel (b), we extend the line joining $(x_j,e^1_{j,n})$ with $(x_{j+1},e^1_{j,n+1})$. To the left of $x_{j}$ this dotted line lies above the dashed line with slope $d^1_{j,n}$ passing through $(x_j, e^1_{j,n})$ which in turn lies above the circles.
}
\label{fig:mixedinter}
\end{figure}

\end{proof}

We extend the trading strategy of Definition~\ref{def:tradingstrategy} to the present context in two ways. First, we consider the European option payoffs $e^\delta_{j,n}$ to be made up of call options with strikes $k_j \in \sX$. Then the payoff from a portfolio which has value $e_{j,n}$ at $x_j$ is $\bar{e}_n(x)$ at $x$, for $x \in [0,x_J]$. Second we use the hedge ratios $\tilde{d}_n^{\delta}$ defined on $[0,x_J]$ rather that $d^\delta_n$.

If ${\bf B}$ denotes the Arrow-Debreu style payoff in Definition~\ref{def:tradingstrategy} then if $\bar{b}_n$ is the linear interpolation of $b_{j,n}$ on $[0,x_J]$ then we must be able to write $\bar{b}_n(x) = b_{0,n} + \sum_{0 \leq j < J} \beta_{j,n}(x-x_j)^+$. The payoff of the strategy becomes $\sum_{1 \leq n \leq N} \bar{b}(X_{t_n})$ and the cost is $\sum_{1 \leq n \leq N}( b_{0,n} + \sum_{0 \leq j \leq J} \beta_{j,n} c_{j,n})$.

Let $\sH^{[0,x_J],\sT}(a)$ be the set of super-replicating semi-static strategies which super-replicate for all exercise times and for all price paths with $x_{t_n} \in [0,x_J]$.

\begin{proposition}
\label{prop:linsuper}
If the quintuple $({\bf E^1}, {\bf E^2}, {\bf D^1}, {\bf D^2}, {\bf V})$ is feasible for ${\bf L^{\sX,\sT}_H}$ then the trading strategy in Definition~\ref{def:tradingstrategy}, extended as above, super-replicates the American claim along all paths with $X_{t_n} \in [0,x_J]$ for $1 \leq n \leq N$.
\end{proposition}

\begin{proof}
Suppose the American claim is exercised at $\tau \in \sT$, and that $X= (X_{t_n})_{0 \leq n \leq N}$ follows the path $(s_0, y_1, \ldots, y_N)$ with $y_i \in [0,x_J]$. The terminal payoff $\sG_T= \sG_T(y_1, \ldots, y_N)$ from the strategy is
\begin{eqnarray*} \sG_T & = & \sum_{n=1}^N  (\bar{e}^{1}_{n}(y_n) + \bar{e}^2_{n}(y_n))  + \bar{v}_N(y_N) \\
&& \hspace{5mm}+ \sum_{1}^{\sN(\tau)-1} (y_{n+1}- y_n) \tilde{d}^1_{n}(y_n)+ \sum_{\sN(\tau)}^{N-1}  (y_{n+1}- y_n) \tilde{d}^2_{n}(y_n) \\
& = & \bar{e}^2_1(y_1) + \bar{e}^1_N(y_N) + \{ \bar{v}_{\sN(\tau)}(y_\tau) - a(y_\tau,\tau) \} \\
&& \hspace{5mm}+\sum_{1}^{\sN(\tau)-1} \left\{ \bar{e}^1_n(y_n) + \bar{e}^2_{n+1}(y_{n+1}) + (y_{n+1}- y_n) \tilde{d}^1_{n}(y_n) \right\} \\
&& \hspace{5mm}+\sum_{\sN(\tau)}^{N-1} \left\{ \bar{e}^1_n(y_n) + \bar{e}^2_{n+1}(y_{n+1}) + (y_{n+1}- y_n) \tilde{d}^2_{n}(y_n) - \bar{v}_n(y_n) + \bar{v}_{n+1}(y_{n+1}) \right\}  \\
&& \hspace{5mm} + a(y_\tau,\tau).
\end{eqnarray*}
The first two elements in the second expression are zero. The third is non-negative since it is non-negative on $\sX$, $\bar{v}_n$ is a linear interpolation between members of $\sX$ and $a$ is convex in the first argument. The fourth and fifth terms are non-negative by Proposition~\ref{prop:linearized}. Hence it follows that $\sG_T(y_1, \ldots, y_N) \geq a(y_\tau,\tau)$ and hence for every possible path in $[0,x_J]^N$, and for every possible time $\tau \in \sT$ the strategy super-replicates.
\end{proof}

It follows that there is an analogue of Theorem~\ref{thm:main1} for this setting, but we state it in a slightly more general form at the end of the next section.

\subsection{A super-hedge for price processes on $\R^+ \times \sT$}
\label{ssec:R}

Under our current assumptions, call options with strikes $x_J$ trade at zero price, and it follows that in any model which is consistent with ${\bf C}$, the price process never gets above $x_J$. Our proof of super-replication considered paths which respected this bound. Nonetheless, ideally we would like our super-replicating strategies to super-hedge for all scenarios for the price process and not just those in which $X_{t_n} \leq x_J$. In this section we describe a superhedge which works for all paths, and which costs the same as the cheapest super-hedge from the previous section. This strategy involves initial purchases of calls with strike $x_J$ which are available at zero price.

\begin{assumption}
\label{ass:Aass} Time is discrete and takes values in the finite set $\sT_0$. The price process $X = (X_t)_{t \in \sT_0}$ takes values in $\R_+$.
The American option payoff $a:\R^+ \times \sT \mapsto \R$ is such that, in addition to being positive and convex in its first argument it also has at most linear growth:  $\lim_{x \uparrow \infty} a(x,t_n)/x < R$ for each $t_n \in \sT$.
\end{assumption}

When prices takes values in $\R^+$ and not just in $\sX$ we add to the definition of a semi-static strategy the requirement that $\Theta^1$ and $\Theta^2$ are bounded. Then weak duality still holds.

Let $\sS^{\R^+,\sT}(a)$ be the set of super-replicating semi-static strategies which super-replicate for all exercise times and for all price paths with $x_{t_n} \in \R^+$.

\begin{definition}
\label{def:add}
In addition to the portfolio holdings/strategy implicit in the quintuple $({\bf E^1}, {\bf E^2}, {\bf D^1}, {\bf D^2}, {\bf V})$ and described in Definition~\ref{def:tradingstrategyEEDDV} and extended in the observations before Proposition~\ref{prop:linearized}, add the payoff
\[ \sum_{1 \leq n \leq N} \beta_{J,n} (X_{t_n} - x_J)^+ \]
by adding $\beta_{J,n}$ calls with maturity $t_n$ and strike $x_J$ for each $n$. Here
\begin{eqnarray*} \beta_{J,n} & = & I_{ \{ 2 \leq n \leq N \} } \left[ \left( \inf_{0 \leq x \leq x_J} \tilde{d}^1_{n-1}(x) \right)^- + \left( \inf_{0 \leq x \leq x_J} \tilde{d}^2_{n-1}(x) + R \right)^- \right] \\
&& \hspace{5mm} +  I_{ \{1 \leq n \leq N-1 \}} \left( d^1_{J,n} + d^2_{J,n} + R \right) . \end{eqnarray*}
Payoffs from these additional options maturing in the money are held until time $T$. The additional payoff is costless, since $c_{J,n}=0$ for all $n$.
\end{definition}

Given $({\bf {\bar{E}^1}, {\bf \bar{E}^2}, {\bf \tilde{D}^1}, \bf \tilde{D}^2}, {\bf \bar{V}})$ defined on $[0,x_J]$, extend the definitions to $\R^+$ by
\[ \bar{e}^\delta_{n}(y) = \bar{e}^\delta_{n}(y) \hspace{10mm} {d}^\delta_n(y)  = {d}^\delta_{J,n} \hspace{10mm} \bar{v}_n(y) =   \bar{v}_n(x_J) + R(y-x_J) \]
for $y>x_J$, $1 \leq n \leq N$ and $\delta = 1,2$.

\begin{proposition}
\label{prop:linsuper2}
If the quintuple $({\bf E^1}, {\bf E^2}, {\bf D^1}, {\bf D^2}, {\bf V})$ is feasible for ${\bf L^{\sX,\sT}_H}$ then the trading strategy in Definition~\ref{def:add} super-replicates the American claim along all paths with $x_{t_n} \in \R_+$ for $1 \leq n \leq N$.
\end{proposition}

\begin{proof}
If the price process takes values $(s_0, y_1, \ldots, y_N)$ with $y_m \in \R^+$  and if $\tau \in \{ t_1, \ldots t_N \}$ we find the terminal payoff $\sG_T = \sG_T(y_1, \ldots, y_N)$ is given by
\begin{eqnarray*} \sG_T & = & \sum_{n=1}^N  (\bar{e}^{1}_{n}(y_n) + \bar{e}^2_{n}(y_n))  + \bar{v}_N(y_N) + R(y_N - x_J)^+ \\
&& \hspace{5mm} + \sum_{1}^{\sN(\tau)-1} (y_{n+1}- y_n) \tilde{d}^1_{n}(y_n)+ \sum_{\sN(\tau)}^{N-1}  (y_{n+1}- y_n) \tilde{d}^2_{n}(y_n) \\
&& \hspace{5mm} + \sum_1^{N-1} \left[ d^1_{J,n} (y_n - x_J)^+ + \left( \inf_{0 \leq x \leq x_J} \tilde{d}^1_n(x) \right)^- (y_{n+1} - x_J)^+ \right] \\
&& \hspace{5mm} + \sum_1^{N-1} \left[ (d^2_{J,n}+R) (y_n - x_J)^+ + \left( \inf_{0 \leq x \leq x_J} \tilde{d}^2_n(x) + R \right)^- (y_{n+1} - x_J)^+ \right] \\
& \geq & \bar{e}^2_1(y_1) + \bar{e}^1_N(y_N) + \{ \bar{v}_{\sN(\tau)}(y_\tau)  - a(y_\tau,\tau) \} \\
&& \hspace{5mm}+\sum_{1}^{\sN(\tau)-1} \left\{ \bar{e}^1_n(y_n) + \bar{e}^2_{n+1}(y_{n+1}) + (y_{n+1}- y_n) \tilde{d}^1_{n}(y_n) + d^1_{J,n}(y_n - x_J)^+ \right. \\
&& \hspace{20mm} \left. + \left( \inf_{0 \leq x \leq x_J} \tilde{d}^1_n(x) \right)^- (y_{n+1} - x_J)^+ \right\} \\
&& \hspace{5mm}+\sum_{\sN(\tau)}^{N-1} \left\{ \bar{e}^1_n(y_n) + \bar{e}^2_{n+1}(y_{n+1}) + (y_{n+1}- y_n) \tilde{d}^2_{n}(y_n)  + [d^2_{J,n}+R](y_n - x_J)^+ \right. \\
&& \hspace{20mm} \left. - \bar{v}_n(y_n) + \bar{v}_{n+1}(y_{n+1}) +  \left( \inf_{0 \leq x \leq x_J} \tilde{d}^2_n(x) + R \right)^- (y_{n+1} - x_J)^+ \right\}\\
&& \hspace{5mm} + a(y_\tau,\tau).
\end{eqnarray*}
We find that $\sG_T \geq  a(y_\tau,\tau)$ for all times $\tau$ and for all paths provided the terms in the two sums are non-negative. We consider the second of these. Write $y_{n+1}=y$ and $y_n = x$ and consider
\begin{eqnarray*} \Upsilon & = & \bar{e}^1_n(x) + \bar{e}^2_{n+1}(y) + (y- x) \tilde{d}^2_{n}(x)  + [d^2_{J,n}+R](x - x_J)^+ - \bar{v}_n(x)  \\
&&  + \bar{v}_{n+1}(y) + \left[ \inf_{0 \leq x \leq x_J} \tilde{d}^2_n(x) + R \right]^- (y - x_J)^+
\end{eqnarray*}
We can consider four cases according as $x$ or $ y$ lies below or above $x_J$.
If $x,y \leq x_J$ the non-negativity of $\Upsilon$ follows from Proposition~\ref{prop:linearized}. If $x_J < x,y$,
\begin{eqnarray*}
\Upsilon & = & \bar{e}^1_n(x_J) + \bar{e}^2_{n+1}(y_J) - \bar{v}_n(x_J) + \bar{v}_{n+1}(x_J) \\
&& \hspace{5mm} + (x_J-x) {d}^2_{J,n} - R(x - x_J)  + [d^2_{J,n}+R](x - x_J)  \\
&& \hspace{5mm} + (y - x_J) d^2_{J,n} + R(y - x_J) +  \{ \inf_{0 \leq x \leq x_J} \tilde{d}^2_n(x) + R \}^- (y - x_J)^+
\end{eqnarray*}
and $\Upsilon \geq 0$ since each of the three lines in this expression is non-negative.

The other cases follow similarly.
\end{proof}

Denote by $\sP^{\R^+,\sT}(a,{\bf C})$ the highest model based price for the American option over discrete-time models consistent with the call prices for which the price at time $t \in \sT$ is only constrained to be non-negative. Then $\sP^{\R^+,\sT}(a,{\bf C})= \sup_{M \in \sM^{\R^+,\sT}(a,{\bf C})} \phi(M)$.

Let $\sS^{\R^+,\sT}(a)$ denote the space of semi-static strategies which super-replicate the American payoff $a$ along price paths taking values in $\R^+$, and let
$\sH^{\R^+,\sT}(a,{\bf C})= \inf_{(B, \Theta) \in \sS^{\R^+,\sT}(a)}H({\bf B})$ denote the cost of the cheapest super-replicating strategy.

\begin{theorem} We have $\Phi^{\sX,\sT}(a, {\bf C}) = {\sP}^{\R^+,\sT}(a, {\bf C}) = {\sH}^{\R^+,\sT}(a, {\bf C}) = \Psi^{\sX,\sT}(a, {\bf C})$. In particular,
the most expensive model-based price amongst models which are consistent with the observed call prices is attained by a price/regime model in which the price only takes values in $\sX$ (an element of $\sM_2^{\sX,\sT}({\bf C}$). Similarly, there is a super-replicating strategy of the form described in Definition~\ref{def:add} which super-hedges against all exercise times, and along all non-negative paths for which the cost of the strategy is the lowest amongst the class of all super-replicating semi-static strategies.
\end{theorem}

\begin{proof}
The proof of weak duality (Proposition~\ref{prop:wd}) does not use the fact that the price process take values in $\sX$, and so applies in this context. Then
we have
\[ \Phi^{\sX,\sT}(a, {\bf C}) = {\sP}^{\sX,\sT}(a, {\bf C}) \leq {\sP}^{\R^+,\sT}(a, {\bf C}) \leq {\sH}^{\R^+,\sT}(a, {\bf C}) = {\sH}^{\sX,\sT}(a, {\bf C}) = \Psi^{\sX,\sT}(a, {\bf C}) \]
where we use Theorem~\ref{thm:main1} for the outer two equalities, the set inclusion
$\sM^{\sX,\sT}({\bf C}) \subseteq \sM^{\R^+,\sT}({\bf C})$ for the first inequality, Proposition~\ref{prop:wd} for the second inequality and
Proposition~\ref{prop:linsuper2} to conclude that ${\sH}^{\R^+,\sT}(a, {\bf C}) = {\sH}^{\sX,\sT}(a, {\bf C})$.
But, $\Phi^{\sX,\sT}(a, {\bf C}) = \Psi^{\sX,\sT}(a, {\bf C})$ since they are the values of a pair of dual linear programmes (Proposition~\ref{prop:primaldualequality}).
\end{proof}

\subsection{Unrestricted exercise times}
\label{ssec:unrestricted}
In the prequel we have assumed that the price process was defined with discrete time-parameter set $\sT_0$, the American payoff was defined on $\R \times \sT$ and that the stopping time was restricted to lie in $\sT$. Now we assume that we are given option prices for a finite set of maturities $t_k \in \sT$ but we want to allow for more general exercise times. First we extend the set of allowable exercise dates to $\sT_0 = \{ 0 \} \cup \sT$. Then we extend the results to allow for exercise at any time $\tau \in \bbT=[0,T]$ under an monotonicity (in time) assumption on the American payoff.

Suppose $a : \R_+ \times \sT_0$ is the payoff function. Define $c^0_{j,1} = (s_0 - x_j)^+$ and $c^0_{j,n} = c^0_{j,n-1}$. Then ${\bf C}^0$ satisfies Assumption~\ref{ass:conditionsonC} and the analysis proceeds exactly as before. In particular, the corresponding primal and dual problems have a solution, the solutions are equal, and they correspond to the highest model based price and cheapest super-replicating strategy.

Now consider the more interesting case in which $\tau$ may take any value in $[0,\bbT]$.
\begin{assumption}
\label{ass:AassU} Time is continuous and takes values in the set $\bbT = [0,T]$. The price process $X = (X_t)_{t \in \bbT}$ takes values in $\R^+$.
The American option payoff $A:\R^+ \times \bbT \mapsto \R$ is positive, convex in its first argument with $\lim_x A(x,t)/x < R$ for each $t \in [0,T]$ and decreasing in its second argument.
\end{assumption}

We suppose we are given a set of call prices ${\bf C} = c_{j,n}$ for strikes $x_j \in \sX$ and maturities $t \in \sT$.

\begin{definition}
\label{def:modelcts}
$\sM^{\R^+, \bbT}({\bf C})$ is the set of continuous-time models (i.e. a filtration $\bbF = (\sF_t)_{0 \leq t \leq T}$ a probability measure $\Prob$ and a stochastic process $X= (X_{t})_{0 \leq t \leq T}$ taking values in $\R^+$) such that $X_0 = s_0$, and
\begin{enumerate}
\item the process $X$ is consistent with {\bf C} in the sense that $\E[(X_{t_n}-x_j)^+] = c_{j,n}$ for all $x_j \in \sX$ and all $t \in \sT$ ;
\item $X$ is a right-continuous $(\Prob,\bbF)$-martingale.
\end{enumerate}
We say such a model is {\em consistent} with the call prices ${\bf C}$.
\end{definition}

For $M \in \sM^{\R^+,\bbT}$ define $\phi^A(M) = \sup_\tau \E^M[A(X_\tau,\tau)]$ where $\tau$ takes values in $\bbT$.

If price processes are defined on $\bbT = [0,T]$ and exercise is allowed at any time, then we need to allow for more general dynamic hedging strategies than those given in Definition~\ref{def:tradingstrategy}.
In particular the set of admissible dynamic strategies must allow for piecewise constant positions in the stock with rebalancings at the times $t \in \sT$ {\em and} at the exercise time $\rho$, in which the size of the position at time $t$ depends on $(x_{t_1 \wedge t}, \ldots , x_{t_{N} \wedge t})$ (before $\rho$) and $(x_{t_1 \wedge t}, \ldots , x_{t_{N} \wedge t}, x_{\rho}, \rho)$ (after exercise).

Let $\sN(t) = \min \{n : t_n \geq t\}$ and suppose $\Theta$ is of this form.
Then the gains from trade from the dynamic hedging strategy is
\begin{eqnarray}
 \sG_T^{\bf \Theta} &  =  & \sum_{n=1}^{\sN(\sigma)-1} \Theta^1_{t_n}(x_{t_1},\ldots, x_{t_n})(x_{t_{n+1}} - x_{t_n})
+ \Theta^2_{\rho}(x_{t_1},\ldots, x_{t_{\sN(\rho)-1}},x_\rho,\rho)(x_{t_{\sN(\rho)}} - x_{\sigma}) \nonumber \\
&& \hspace{5mm} + \sum_{n=\sN(\rho)+1}^{N-1} \Theta^2_{t_n}(x_1,\ldots, x_\rho, \ldots, x_{t_n}, \sigma)(x_{t_{n+1}} - x_{t_n}) \label{eqn:Gdef} \end{eqnarray}

Let $\sS^{\R^+,\bbT}(A)$ denote the space of super-replicating semi-static strategies such that $\sG_T \geq A(x_\sigma,\sigma)$ for all non-negative price paths on $[0,T]$, and all exercise times taking values in $[0,T]$.

In continuous time our assumption is that the space of admissible semi-static strategies dynamic strategies includes gains from trade of the form in (\ref{eqn:Gdef}). In particular, the definition of a semi-static strategy must include the possibility of a rebalancing of the dynamic hedge at the moment of exercise. The space of admissible strategies may be larger, but any admissible strategy must have the twin properties that the gains from trade $\sG^{\bf \Theta}_T = (\Theta \cdot x)$ can be defined pathwise, and that $\E^M[ \sG^{\bf \Theta}_T ] \leq 0$ under any model. The latter is required to rule out doubling strategies. If this is the case then, we have weak duality $\sup_{M \in \sM^{\R^+,\bbT}({\bf C})} \phi^A(M) \leq \inf_{({\bf B},\Theta) \in \sS^{\R^+,\bbT}(A)} H({\bf B})$ as in Proposition~\ref{prop:wd}.

The intuition behind the following theorem is that European option prices determine the range of price movements between successive maturities $t_{n-1}$ and $t_n$, but they say nothing about when these price movements will occur. Since the payoff function is decreasing in $t$, the American option price is highest in a model in which the price movements occur at the beginning of each interval $(t_{n-1},t_n]$. In this way the option holder benefits from the convexity of $A$ without losing from the decline in time.

\begin{theorem}
Suppose the option payoff is given by $A(x,t)$ where $A:\R^+ \times \bbT \mapsto \R$ satisfies Assumption~\ref{ass:AassU}.
Define $a(x,t_k) = \lim _{t \downarrow t_{k-1}} A(x,t) = A(x, t_{k-1}+)$. Assume that the conditions on the space of dynamic strategies are such that weak duality holds.
Then
\[ \Phi^{\sX,\sT}(a,{\bf C}) = \sup_{M \in \sM^{\R^+,\bbT}({\bf C})} \phi^A(M) = \inf_{({\bf B},\Theta) \in \sS^{\R^+,\bbT}(A)} H_{{\bf C}}({\bf B}) = \Psi^{\sX,\sT}(a, {\bf C}). \]
In particular, the supremum of the American option price with payoff $A$ over models which are consistent with the call prices ${\bf C}$ and the cost of the cheapest super-replicating semi-static strategy are both equal to $\Psi^{\sX, \sT}(a,{\bf C})$.
\end{theorem}

\begin{proof}

The payoff $a$ is positive, convex in its first argument, and $\lim_x A(x,t)/x < R$. In particular $a$ satisfies all the assumptions on the payoff function required for the results of the previous subsection to hold. We will argue that we can find a super-replicating strategy for the payoff function $A$ with associated super-hedging price $\Psi = \Psi(a,{\bf C})$ and that
there is a sequence of models which are consistent with the call prices ${\bf C}$
for maturities $t \in \sT$ for which the associated prices for the American option with payoff $A$ converge to $\Phi(a, {\bf C})$. By weak duality it will follow that we have solved both the primal and dual problems and that the superhedge outlined in the previous paragraph is the cheapest super-hedge for American options with unrestricted exercise dates.

First we show how to extend the notion of a superhedging strategy to processes (and exercise times) in continuous time.
Recall Definition~\ref{def:tradingstrategy} and suppose we are
given three $(J+1) \times N$ matrices ${\bf E^1}$, ${\bf E^2}$ and ${\bf V}$ and two $(J+1)\times(N-1)$ matrices ${\bf D^1}$
and ${\bf D^2}$. In addition to the elements of the trading strategy described in Definition~\ref{def:tradingstrategy} (using ${\bf \tilde{D}^1}$ as $\bf{\tilde{D}^2}$ as in Proposition~\ref{prop:linearized}), add that if the American option is exercised at a time $\tau \in [0,T]$ with $t_m < \tau < t_{m+1}$, and the asset price is $X_\tau$ then take a short position of $a'_+(X_\tau, t_{m+1})$ units of stock (financed by borrowing), and liquidate this position at $t_{m+1}$. Note that since $a$ is convex the right-derivative $a'_+$ is well defined everywhere.

If the option is exercised at a time $\tau \in \sT$ then the strategy super-replicates exactly as before.
Otherwise, the effect of this additional element of the strategy is to add a term $-a'_+(X_\tau, t_n)(X_{t_{m+1}} - X_{\tau})$, relative to the expressions in Proposition~\ref{prop:linsuper} to the payoff so that it becomes (we add and subtract $a(y_\tau,t_{m+1})$ and $a(y_{t_{m+1}},t_{m+1})$ rather than $a(y_\tau,\tau)$)
\begin{eqnarray*} \sG_T
& = & \bar{e}^2_1(y_1) + \bar{e}^1_N(y_N) + \{ \bar{v}_{m+1}(y_{t_{m+1}}) - a(y_{t_{m+1}},t_{m+1}) \} \\
&& \hspace{5mm}+ \{ a(y_{t_{m+1}},t_{m+1}) - a(y_\tau,t_{m+1}) - a'_+(y_\tau, t_{m+1})(y_{t_{m+1}} - y_{\tau})  \} \\
&& \hspace{5mm}+\sum_{1}^{m} \left\{ \bar{e}^1_n(y_n) + \bar{e}^2_{n+1}(y_{n+1}) + (y_{n+1}- y_n) \tilde{d}^1_{n}(y_n) \right\} \\
&& \hspace{5mm}+\sum_{m+1}^{N-1} \left\{ \bar{e}^1_n(y_n) + \bar{e}^2_{n+1}(y_{n+1}) + (y_{n+1}- y_n) \tilde{d}^2_{n}(y_n) - \bar{v}_n(y_n) + \bar{v}_{n+1}(y_{n+1}) \right\}  \\
&& \hspace{5mm} + a(y_\tau,t_{m+1})
\end{eqnarray*}

Since $a$ is convex we have $a(y_{t_{m+1}},t_{m+1}) \geq a(y_\tau,t_{m+1}) + a'_+(y_\tau, t_{m+1})(y_{t_{m+1}} - y_{\tau})$ and $\sG_T \geq a(y_\tau,t_{m+1}) = A(y_\tau, t_{m}+) \geq A(y_\tau,\tau)$ since $A$ is decreasing in its second argument. Hence we have a family of super-replicating strategies; minimising over the cost of such strategies gives
\[ \inf_{({\bf B},\Theta) \in \sS^{\R^+,\bbT}(A)} H({\bf B}) \leq \Psi^{\sX,\sT}(a, {\bf C}). \]

Now we turn to the pricing problem.
Suppose $\min_{1 \leq n \leq N} \{ t_n - t_{n-1} \} = \epsilon_0$. Let $(X,\Delta)$ be a process with time-parameter set $\sT$ taking values in $\sX \times \{1,2\}$, and such that $\E[(X_{t_n}-x_j)^+] = c_{j,n}$. The model-based price of the American option with payoff $A$ in this model is $\sum_{j,n} A(x_j,t_n) f_{j,n}$.

Choose $(X, \Delta)=(X^{a, {\bf C}}, \Delta^{a, {\bf C}})$ so that it is the discrete-time process associated with the optimiser in Linear Program~\ref{LP:primal} for the payoff $a$.
Let $\sF_{t_n} = \sigma( X_{t_m}, \Delta_{t_m}; m \leq n)$: extend the time-parameter set of the filtration to $[0,T]$ by setting $\sF_t = \cup_{n: t_n \leq t} \sF_{t_n}$.

For $\epsilon \in (0, \epsilon_0)$ define $\bbF^\epsilon = (\sF^\epsilon_t)_{0 \leq t \leq T}$ by $\sF^\epsilon_t  = \sF_0$ for $t < \epsilon$, $\sF^\epsilon_t = \sF_{t_n}$ for $\epsilon + t_{n-1} \leq t < (\epsilon + t_n)$ and $n<N$ and $\sF^\epsilon_t  = \sF_T$ for $\epsilon + t_{N-1} \leq t \leq T$. Define a family of piecewise constant, right-continuous, continuous-time, bivariate processes $(Y^\epsilon, \Gamma^\epsilon)$ by
\[ (Y^\epsilon_t, \Gamma^\epsilon_t) = \left\{ \begin{array}{ll} (s_0,1) & 0 \leq t < \epsilon   \\\
                                        (X_{t_n}, \Delta_{t_n}) & \epsilon + t_{n-1} \leq t < \epsilon + t_n \\
                                        (X_{T}, \Delta_{T}) & \epsilon + t_{N-1} \leq t \leq T   \end{array} \right. \]
Then $(Y^\epsilon, \Gamma^\epsilon)$ is a $\bbF^\epsilon$-stochastic process obtained from $(X, \Delta)$ by changing the jump times from $\{ t_1, t_2, \ldots t_n=T\}$ to $\epsilon, \epsilon+ t_1, \ldots,   \epsilon + t_{N-1}$ (and extending the time domain to $[0,T]$ by making the process constant between these jump-times). Moreover we have an identity in law $\sL(Y_{t_n}) = \sL(X_{t_n})$ by construction, so that $Y$ is consistent with the prices of traded calls (at the traded maturities $t \in \sT$). The model-based price of the American option with payoff $A$ if the asset price/regime pair is described by $(Y^\epsilon, \Gamma^\epsilon)$ (which is obtained by using a strategy of exercising as soon as the regime process $\Gamma^\epsilon$ has jumped to two) is
\[ \sum_{1 \leq n \leq N} \sum_{0 \leq j \leq J} A(x_j, t_{n-1}+ \epsilon) f_{j,n} \]
which increases to $\sum_{1 \leq n \leq N} \sum_{0 \leq j \leq J} a(x_j,t_{n}) f_{j,n} = \Phi^{\sX,\sT}(a, {\bf C})$ as $\epsilon$ decreases to zero.
Hence
\[ \sup_{M \in \sM^{\R^+,\bbT}({\bf C})} \phi^A(M)  \geq  \Phi^{\sX, \sT}(a, {\bf C}). \]

The proof is complete since $\Phi^{\sX, \sT}(a, {\bf C}) = \Psi^{\sX, \sT}(a, {\bf C})$.
\end{proof}

\section{Unbounded stock prices, and no call with zero price.}
\label{sec:nocall}

In the previous sections we solved for the most expensive model and the cheapest super-replicating strategy under the restriction that there is a large strike at which the associated call price is zero. Our goal in this section is to relax this assumption, under a slight strengthening of the other elements of Assumption~\ref{ass:conditionsonC}, so that the inequalities become strict.

We return to the case of discrete time, although the results can be extended to continuous time exactly as in Section~\ref{ssec:unrestricted}.

\begin{assumption}
\label{ass:AassNCZ} Time is discrete and takes values in the finite set $\sT_0$ and the exercise time of the option is restricted to lie in $\sT$. The price process $X = (X_t)_{t \in \sT_0}$ takes values in $\R^+$.
The American option payoff $a:\R^+ \times \sT \mapsto \R$ is such that $a$ is positive and convex in its first argument. It also has at most linear growth:  $\lim_{x \uparrow \infty} a(x,t_n)/x < R$ for each $t_n \in \sT$.
\end{assumption}

Again, we assume that there is a finite family of call options traded on the market, one for each pair of strike in $\sK$ and maturity in $\sT$, and again we consider the stock as a call with zero strike. The prices of these calls are written in matrix form as ${\bf C}$. We assume:
\begin{assumption}
\label{ass:strongconditionsonC}
\item The set of option prices has the following properties:
\begin{itemize}
\item For $1 \leq n \leq N$, $s_0 = c_{0,n} > c_{1,n} > c_{2,n} >  c_{J,n} > 0$.
\item For $1 \leq n \leq N$,
\( 1 > \frac{c_{0,n} - c_{1,n}}{x_1} > \frac{c_{1,n} - c_{2,n}}{x_2 -x_1} > \cdots > \frac{c_{J-1,n} - c_{J,n}}{x_J -x_{J-1}} > 0\).
\item For $1 \leq n \leq N-1$, and for $1 \leq j \leq J$,
$c_{j, n+1} > c_{j,n}$.
\end{itemize}
\end{assumption}

Recall the definition of $p_{j,n}$ in (\ref{eq:pdef}).
Introduce the $(J+2) \times N$ matrix $\bf \hat{P}$ via $\hat{p}_{j,n} = p_{j,n}$ for $0 \leq j \leq J$ and $\hat{p}_{J+1,n} = c_{J,n}$. Observe that $\sum_{0\leq j \leq J} \hat{p}_{j,n} = 1 < \sum_{0\leq j \leq J+1} \hat{p}_{j,n}$. Further, given a vector $(h_0, h_1, \ldots h_J)$ and a final element $h_{J+1}$ define the {\em extended linear interpolation} $\bar{h}: \R^+ \mapsto \R$ of $h$ by
\[ \bar{h}(x) = \left\{ \begin{array}{ll} h_j   & x=x_j \in \sX \\
                                         \frac{x-x_j}{x_{j+1}-x_j} h_{j+1} + \frac{x-x_j}{x_{j+1}-x_j} h_{j+1} & x< x_J, x_j < x < x_{j+1} \\
                                         h_J + (x-x_J) h_{J+1} & x>x_J. \end{array} \right.
                                         \]

Let $h$ be a vector $h = (h_0,  h_1, \ldots,  h_J, h_{J+1})$ and let $\bar{h}$ be the extended linear interpolation  of $h$. We can consider $\bar{h}(X_{t_n})$ as a payoff of a European option with maturity $t_n$.

\begin{lemma}
The cost of the claim with maturity $t_n$ and payoff $\bar{h}(X_{t_n})$ is $\sum_{0 \leq j \leq J+1} h_j \hat{p}_{j,n}$.
\end{lemma}

\begin{proof}
$\bar{h}$ is piecewise linear with kinks at elements of $\sX$. In particular, we can write $\bar{h}$ as a sum of call payoffs:
\begin{eqnarray*}
\bar{h}(x) & = & h_0 + \frac{h_1 - h_0}{x_1} x + \sum_{1 \leq j \leq J-1} \left[ \frac{h_{j+1}-h_j}{x_{j+1}-x_j} - \frac{h_{j}-h_{j-1}}{x_j-x_{j-1}} \right] (x-x_j)^+ \\
           &&   \hspace{10mm}         + \left[ h_{J+1} - \frac{h_{J}-h_{J-1}}{x_J-x_{J-1}} \right] (x-x_J)^+ .
\end{eqnarray*}
The cost of this portfolio is
\begin{eqnarray*}
\lefteqn{h_0 + \frac{h_1 - h_0}{x_1} s_0 + \sum_{1 \leq j \leq J-1} \left[ \frac{h_{j+1}-h_j}{x_{j+1}-x_j} - \frac{h_{j}-h_{j-1}}{x_j-x_{j-1}} \right] c_{j,n} + \left[ h_{J+1} - \frac{h_{J}-h_{J-1}}{x_J-x_{J-1}} \right] c_{J,n}} \\
& = & h_0 \left[ 1 - \frac{s_0}{x_1} + \frac{c_{1,n}}{x_1} \right] +  \sum_{1 \leq j \leq J-1} h_j \left[ \frac{c_{j-1,n}-c_{j,n}}{x_{j}-x_{j-1}} - \frac{c_{j,n}-c_{j+1,n}}{x_{j+1}-x_{j}} \right]  \\
&& \hspace{10mm} + h_{J} \left[ \frac{c_{J-1,n}-c_{J,n}}{x_{j}-x_{j-1}} \right] + h_{J+1}c_{J+1,n} \\
& = & \sum_{0 \leq j \leq J+1} h_j \hat{p}_{j,n}
\end{eqnarray*}
\end{proof}

The above lemma motivates following linear program:

\begin{LP}
\label{LP:dualinfinityA}
The hedging problem ${\bf L^{\sX,\infty,\sT}_H}$ is to: \\
find the three $(J+2) \times N$ matrices ${\bf E^1}$, ${\bf E^2}$ and ${\bf V}$ and the two $(J+2)\times(N-1)$ matrices ${\bf D^1}$
and ${\bf D^2}$ which minimise
\[ \psi = \sum_{0 \leq j \leq J, 1 \leq n \leq N} (e^1_{j,n} + e^2_{j,n})\hat{p}_{j,n} + \sum_{0 \leq j \leq J+1} v_{j,N} \hat{p}_{j,N}  \]
subject to ${\bf V} \geq 0$, and
\begin{enumerate}
\item[(i)] $v_{j,n} \geq a(x_j,t_n)$; $0 \leq j \leq J$, $1 \leq n \leq N$;  \\
                $v_{J+1,n} \geq \lim_{x \uparrow \infty} a(x,t_n)/x$ for $1 \leq n \leq N$.
\item[(ii)] $e^1_{j,n} + e^2_{k,n+1} + (x_k-x_j)d^1_{j,n} \geq 0$; $0 \leq j,k \leq J$, $1 \leq n \leq N-1$; \\
            $e^1_{J+1,n} -d^1_{J+1,n} \geq 0$; $0 \leq k \leq J$, $1 \leq n \leq N-1$; \\
            $e^2_{J+1,n+1} + d^1_{j,n} \geq 0$; $0 \leq j \leq J$, $1 \leq n \leq N-1$; \\
            $e^1_{J+1,n} + e^2_{J+1,n+1} \geq 0$; $0 \leq j \leq J$, $1 \leq n \leq N-1$.
\item[(iii)] $e^1_{j,n} + e^2_{k,n+1} + (x_k-x_j)d^2_{j,n} - v_{j,n} + v_{k,n+1} \geq 0$; $0 \leq j,k \leq J$, $1 \leq n \leq N-1$; \\
          $e^1_{J+1,n} - d^2_{J+1,n} - v_{J+1,n} \geq 0$; $0 \leq k \leq J$, $1 \leq n \leq N-1$; \\
          $e^2_{J+1,n+1} + d^2_{j,n} + v_{J+1,n+1} \geq 0$; $0 \leq j \leq J$, $1 \leq n \leq N-1$; \\
          $e^1_{J+1,n} + e^2_{J+1,n+1} - v_{J+1,n} + v_{J+1,n+1} \geq 0$; $1 \leq n \leq N-1$,
\end{enumerate}
and $e^1_{j,N}=e^2_{j,1}=0$.
Let the optimum value be given by $\Psi^{\sX,\infty,\sT} = \Psi^{\sX,\infty,\sT}(a, {\bf C})$.
\end{LP}

\begin{lemma}
\label{lem:anypath}
Suppose the quintuple $({\bf E^1}, {\bf E^2}, {\bf D^1}, {\bf D^2}, {\bf V})$ satisfies the feasibility conditions of Linear Program~\ref{LP:dualinfinityA}.
For fixed $n$ let $\bar{e}^{1}_n$, $\bar{e}^{1}_n$ and $\bar{v}_n$ be the extended linear interpolations of $(e^{1}_{j,n})_{0 \leq j \leq J+1}$, $(e^{2}_{j,n})_{0 \leq j \leq J+1}$ and $(v_{j,n})_{0 \leq j \leq J+1}$ and let $\tilde{d}^\delta_{n}(x)$ be given by (\ref{eq:tilded}) for $0 \leq x \leq x_J $ and for $x>x_J$,
\[ \tilde{d}^{1}_{n}(x) = \tilde{d}^{1}_{n} = \min \{ d^{1}_{J,n}, e^{1}_{J+1,n} \} \hspace{10mm} \tilde{d}^{2}_{n}(x) = \tilde{d}^{2}_{n} = \min \{ d^{2}_{J,n}, e^{1}_{J+1,n} - v_{J+1,n} \}. \]
Define also $\bar{e}^1_N(x)=0=\bar{e}^2_1(y)$.

Then we have
\begin{eqnarray*}
\bar{v}_n(x)  \geq  a(x,t_n)   & \; \; & x \geq 0, 1 \leq n \leq N \\
\bar{e}^1_n(x) + \bar{e}^2_{n+1}(y) + (y-x) \tilde{d}^1_n(x) \geq 0 && x,y \geq 0, 1 \leq n < N \\
\bar{e}^1_n(x) + \bar{e}^2_{n+1}(y) + (y-x) \tilde{d}^2_n(x) - \bar{v}_n(x) + \bar{v}_{n+1}(y) \geq 0 && x,y \geq 0, 1 \leq n < N.
\end{eqnarray*}
\end{lemma}

\begin{proof}
The inequality for $\bar{v}$ follows from the convexity of $a$. For the two other inequalities the case of $x,y \leq x_J$ has already been covered in Proposition~\ref{prop:linearized}. So, as an example of the remaining analysis consider $W=\bar{e}^1_n(x) + \bar{e}^2_{n+1}(y) + (y-x) \tilde{d}^1_n(x)$.

For $y \leq x_J < x$ we have
\begin{eqnarray*}
W & = & {e}^1_{J,n} + (x-x_J)e^1_{J+1,n} + \bar{e}^2_{n+1}(y) + (y-x_J + x_J-x)) \tilde{d}^1_n(x) \\
& = & ({e}^1_{J,n} + \bar{e}^2_{n+1}(y) + (y-x_J)\tilde{d}^1_n) + (x-x_J)(e^1_{J+1,n} - \tilde{d}^1_n) \\
&\geq &  ({e}^1_{J,n} + \bar{e}^2_{n+1}(y) + (y-x_J){d}^1_{J,n}) \geq  0
\end{eqnarray*}
where we use $\tilde{d}^1_{n} \leq d^1_{J,n}$ and $\tilde{d}^1_{n} \leq e^1_{J+1,n}$.

For $x \leq x_J < y$ we have
\begin{eqnarray*}
W & = & \bar{e}^1_{n}(x) + {e}^2_{J,n+1} + (y-x_J)e^2_{J+1,n+1} + (y-x_J + x_J-x)) \tilde{d}^1_n(x) \\
& = & (\bar{e}^1_{n}(x) + \bar{e}^2_{J,n+1} + (x_J-x)\tilde{d}^1_n(x)) + (y-x_J)(e^2_{J+1,n+1} + \tilde{d}^1_n(x)) \geq 0
\end{eqnarray*}
since $e^2_{J+1,n+1} + \tilde{d}^1_n \geq \min_{0 \leq j \leq J} \{ e^2_{J+1,n+1} + {d}^1_{j,n} \} \geq 0 $.

Now suppose $x_J < x \leq y$ and note that $\tilde{d}^1_n + e^2_{J+1,n+1} = \min \{ {d}^1_{J,n} + e^2_{J+1,n+1} ,   e^1_{J+1,n} + e^2_{J+1,n+1}\} \geq 0$:
\begin{eqnarray*}
W & = & {e}^1_{J,n} + (x-x_J)e^1_{J+1,n} + {e}^2_{J,n+1} + (y-x+x-x_J) {e}^2_{J+1,n+1} + (y-x) \tilde{d}^1_n(x) \\
& = & ({e}^1_{J,n} + {e}^2_{J,n+1}) + (x - x_J) (e^1_{J+1,n} + {e}^2_{J+1,n+1}) + (y-x)(e^2_{J+1,n+1} + \tilde{d}^1_n) \geq  0.
\end{eqnarray*}
Finally for $x_J < y < x$, and using $\tilde{d}^1_{n} \leq e^1_{J+1,n}$,
\begin{eqnarray*}
W & = & {e}^1_{J,n} + [(y-x_J+x-y)]e^1_{J+1,n} + {e}^2_{J,n+1} + (y-x_J){e}^2_{J+1,n+1} + (y-x) \tilde{d}^1_n(x) \\
& = & ({e}^1_{J,n} + {e}^2_{J,n+1}) + (y - x_J) (e^1_{J+1,n} + {e}^2_{J+1,n+1}) + (x-y)(e^1_{J+1,n} - \tilde{d}^1_n) \geq  0.
\end{eqnarray*}
\end{proof}

\begin{corollary}
\label{cor:cheapest}
The optimal value $\Psi^{\sX,\infty,\sT} = \Psi^{\sX,\infty,\sT}(a, {\bf C})$ of ${\bf L^{\sX,\infty,\sT}_H}$  exists. The problem can be interpreted as the search for the cheapest semi-static strategy (of a certain class) which super-replicates the American claim for all exercise dates in $\sT$, and along all price paths. We have
\[ \inf_{ ({\bf B}, \Theta) \in \sS^{\R^+,\sT}(a) } H({\bf B}) \leq \Psi^{\sX,\infty,\sT}(a, {\bf C}). \]
\end{corollary}

\begin{proof}
Lemma~\ref{lem:anypath} together with Proposition~\ref{prop:linsuper} imply that any feasible quintuple is associated with a strategy which superreplicates the claim.

Let $({\bf E^1}, {\bf E^2}, {\bf D^1}, {\bf D^2}, {\bf V})$ be given by ${\bf E^1}= {\bf E^2}= {\bf D^1}=0$, together with $v_{j,n}= \max_{1 \leq n \leq N} \max_{0 \leq j \leq J} a(x_j,t_n)$ for $0 \leq j \leq J$, $v_{J+1,n} = R$, $d^2_{j,n}= 0$ for $0 \leq j \leq J$ and $d^2_{J+1,n} = -R$.
Then this quintuple is feasible.

Note also that the objective function is bounded below: any strategy which superreplicates the claim for unbounded price paths, also super-replicates the claim for paths constrained to lie in $\sX$.
\end{proof}

The dual ${\bf L^{\sX,\infty,\sT}_P}={\bf L^{\sX,\infty,\sT}_P}({\bf C})$ of the hedging linear program ${\bf L^{\sX,\infty,\sT}_H}$ is the following linear program.
If $M =({\bf F}, {\bf G^1}, {\bf G^2})$ we can define
\begin{equation}
 \label{eq:phiM}
\phi_*(M) = \sum_{0 \leq j \leq J} \sum_{ 1 \leq n \leq N} a(x_j, t_n) f_{j,n} + \sum_{1 \leq n \leq N} f_{J+1,n} \lim_{x \uparrow \infty} \frac{a(x,t_n)}{x} \end{equation}
and the aim of ${\bf L^{\sX,\infty,\sT}_P}$ is to minimise $\phi_*(M)$ over feasible models.
Note that the definition of $\phi_*(M)$ in (\ref{eq:phiM}) makes sense even if the model $M$ is not feasible.

\begin{LP}
\label{LP:primalA}
The pricing problem ${\bf L^{\sX,\infty,\sT}_P}$ is to: \\
find the $(J+2) \times N$ matrix $\bf F$ and the two $(J+2)\times(J+2)\times(N-1)$ matrices ${\bf G^1}$
and ${\bf G^2}$ which maximise
\[ \sum_{0 \leq j \leq J} \sum_{ 1 \leq n \leq N} a(x_j, t_n) f_{j,n} + \sum_{1 \leq n \leq N} f_{J+1,n} \lim_{x \uparrow \infty} \frac{a(x,t_n)}{x} \]
subject to ${\bf F} \geq 0$, ${\bf G^1}\geq 0$, ${\bf G^2}\geq 0$, and
\begin{enumerate}
\item[(a)] $\sum_{0 \leq k \leq J} (g^1_{j,k,n} + g^2_{j,k,n}) = \hat{p}_{j,n}$; \hspace{10mm} $0 \leq j \leq J$, $1 \leq n \leq N-1$. \\
           $\sum_{0 \leq k \leq J+1} (g^1_{J+1,k,n} + g^2_{J+1,k,n}) = \hat{p}_{J+1,n}$; \hspace{10mm} $1 \leq n \leq N-1$.
\item[(b)] $\sum_{0 \leq i \leq J} (g^1_{i,j,n-1} + g^2_{i,j,n-1}) = \hat{p}_{j,n}$; \hspace{10mm} $0 \leq j \leq J$, $2 \leq n \leq N$. \\
           $\sum_{0 \leq i \leq J+1} (g^1_{i,J+1,n-1} + g^2_{i,J+1,n-1}) = \hat{p}_{J+1,n}$; \hspace{10mm} $2 \leq n \leq N$.
\item[(c)] $\sum_{0 \leq k \leq J} (x_k - x_j) g^1_{j,k,n} + g^1_{j, J+1, n}= 0$; \hspace{10mm} $0 \leq j \leq J, 1 \leq n \leq N-1$. \\
           $\sum_{0 \leq k \leq J}  g^1_{J+1,k,n} = 0$; \hspace{10mm} $1 \leq n \leq N-1$.
\item[(d)] $\sum_{0 \leq k \leq J} (x_k - x_j) g^2_{j,k,n} + g^2_{j, J+1, n}= 0$; \hspace{10mm} $0 \leq j \leq J, 1 \leq n \leq N-1$. \\
           $\sum_{0 \leq k \leq J}  g^2_{J+1,k,n} = 0$; \hspace{10mm} $ 1 \leq n \leq N-1$.
\item[(e)] \[ \left\{ \begin{array}{ll} f_{j,1} - \sum_{0 \leq k \leq J} g^2_{j,k,1} \leq 0 & 0 \leq j \leq J \\
                                f_{J+1,1} -\sum_{0 \leq k \leq J+1} g^2_{J+1,k,1} \leq 0 &  \\
                                f_{j,n} - \sum_{0 \leq k \leq J} g^2_{j,k,n} + \sum_{0 \leq i \leq J} g^2_{i,j,n-1} \leq 0 & 0 \leq j \leq J, 1<n<N \\
                                f_{J+1,n} - \sum_{0 \leq k \leq J+1} g^2_{J+1,k,n} + \sum_{0 \leq i \leq J+1} g^2_{i,J+1,n-1} \leq 0 &  1<n<N \\
                                f_{j,N} + \sum_{0 \leq i \leq J} g^2_{i,j,N-1} \leq \hat{p}_{j,N} & 0 \leq j \leq J \\
                                f_{J+1,N} + \sum_{0 \leq i \leq J+1} g^2_{i,J+1,N-1} \leq \hat{p}_{J+1,N} &
                  \end{array} \right. . \]
\end{enumerate}
Let the optimum value be given by $\Phi^{\sX,\infty,\sT}= \Phi^{\sX,\infty,\sT}(a, {\bf C})$.
\end{LP}
Since this program is the dual of ${\bf L^{\sX,\infty,\sT}_H}$ we conclude that $\Phi^{\sX,\infty,\sT}= \Psi^{\sX,\infty,\sT}$.

We would like to interpret the optimal solution to this program as a pricing model in a suitable modification of $\sM^{\sX, \sT}_2({\bf C})$. However, there is no consistent model for which the model price equals $\Phi^{\sX,\infty,\sT}(a, {\bf C})$. Instead,
we give a sequence of consistent models which are based on a finite state space and for which the model based price of the American option converges to $\Phi^{\sX,\infty,\sT}$.

Let $\xi_0 = \max_{ 1 \leq n \leq N } \frac{ (x_J c_{J-1,n} - x_{J-1}c_{J,n})}{(c_{J-1,n} - c_{J,n}) } > x_J$. Take $\xi > \xi_0$. The idea is that we are going to consider a market in which calls at an additional strike $\xi$ are traded for zero price. Let $\sK^\xi = \sK \cup \{ \xi \}$, $\sX^\xi = \sX \cup \{ \xi \}$
and let ${\bf C}^\xi$ be the $(J+2) \times N$ matrix of call prices given by $c^{\xi}_{j,n} = c_{j,n}$ for $0 \leq j \leq J$ and $c^{\xi}_{J+1,n}=0$. The requirement $\xi > \xi_0$ ensures that ${\bf C}^\xi$ satisfies both parts of Assumption~\ref{ass:conditionsonC} and hence we can define the matrix ${\bf P}^\xi$ of probabilities, the spaces $\sM^{\sX^\xi, \sT}({\bf C}^\xi)$ with associated subset $\sM_2^{\sX^\xi, \sT}({\bf C}^\xi)$ and the pricing and hedging linear programs ${\bf L}_{\bf P}^{\sX^\xi, \sT}({\bf C}^\xi)$ and ${\bf L}_{\bf H}^{\sX^\xi, \sT}({\bf C}^\xi)$.

For $1 \leq n \leq N$ we have $p^\xi_{j,n} = \hat{p}_{j,n}$ for $0 \leq j \leq J-1$,
\[ p^\xi_{J,n} = \frac{c_{J-1,n} - c_{J,n}}{x_J - x_{J-1}} - \frac{c_{J,n}}{\xi - x_J } = \hat{p}_{J,n}  - \frac{\hat{p}_{J+1,n}}{\xi - x_J } \]
and $p^\xi_{J+1,n} =  \frac{c_{J,n}}{\xi - x_J } = \frac{\hat{p}_{J+1,n}}{\xi - x_J }$.

Let $M=({\bf F},{\bf G^1},{\bf G^2})$ be feasible for ${\bf L}_{\bf P}^{\sX, \infty, \sT}({\bf C})$. $M$ does not define a model on $\sX \times \sT$ since
$\sum_{0 \leq j \leq J} \hat{p}_{j,n} x_j = s_0 - c_{J,n} < s_0$, so that $M$ could not correspond to a martingale. (If as an alternative we hope that $(\hat{p}_{j,n})_{0 \leq j \leq J+1}$ define the marginal laws then $\sum_{j=0}^{J+1} \hat{p}_{j,n} = 1 + c_{J,n}>1$ and the set $(\hat{p}_{j,n})_{0 \leq j \leq J+1}$ is not a set of probabilities.) Instead, the idea is to use $M$ to define a model $M^\xi$ on $\sX^\xi \times \sT$ which is consistent with the call prices ${\bf C}^\xi$. To this end define $({\bf G^{1,\xi}},{\bf G^{2,\xi}})$ by $g^{\delta,\xi}_{j,n} = g^{\delta}_{j,n}$ for $0 \leq j \leq J-1$, $1 \leq n \leq N$, $\delta = 1,2$ together with
\[ g^{\delta, \xi}_{j,J,n}  = g^{\delta}_{j,J,n} - \frac{g^\delta_{j,J+1,n}}{\xi - x_J} \hspace{10mm} 0 \leq j \leq J-1, 1 \leq n \leq N, \delta = 1,2 \]
\[ g^{\delta, \xi}_{J,J,n}  = g^{\delta}_{J,J,n} - \frac{g^\delta_{J,J+1,n}}{\xi - x_J} - \frac{g^\delta_{J+1,J+1,n}}{\xi - x_J} \hspace{10mm} 1 \leq n \leq N, \delta = 1,2 \]
\[ g^{\delta, \xi}_{j,J+1,n}  =  \frac{g^\delta_{j,J+1,n}}{\xi - x_J} \hspace{10mm} 0 \leq j \leq J, 1 \leq n \leq N, \delta = 1,2\]
\[ g^{\delta, \xi}_{J+1,j,n}  =  \frac{g^\delta_{J+1,j,n}}{\xi - x_J} \hspace{10mm} 0 \leq j \leq J+1, 1 \leq n \leq N, \delta = 1,2, \]
and then define ${\bf F^{\xi}}$ via
$f^{\xi}_{j,n} = f_{j,n}$ for $0 \leq j \leq J-1$, $1 \leq n \leq N$ together with
\[ f^{\xi}_{J,1}  = f_{J,1} - \frac{g^2_{J+1,J+1,1}}{\xi - x_J}  \]
\[ f^{\xi}_{J,n}  = f_{J,n} - \frac{1}{\xi - x_J} \left( g^2_{J+1,J+1,n} - g^2_{J,J+1,n-1} - g^2_{J+1,J+1,n-1} \right) \hspace{10mm} 2 \leq n \leq N-1 \]
\[ f^{\xi}_{J,N}  = f_{J,N} + \frac{1}{\xi - x_J} \left( g^2_{J,J+1,N-1} + g^2_{J+1,J+1,N-1} - c_{J,N}  \right)  \]
\[  f^{\xi}_{J+1,n}  =  \frac{f_{J+1,n}}{\xi - x_J} \hspace{10mm} 1 \leq n \leq N .\]
Note that by (c) and (d) $g^{\delta}_{J+1,k,n} = 0$ for $k \leq J$.
Let $M^\xi=({\bf F^\xi},{\bf G^{1,\xi}},{\bf G^{2,\xi}})$.

\begin{lemma}
\label{lem:verylong}
Suppose $M=({\bf F},{\bf G^1},{\bf G^2})$ satisfies the feasibility conditions (a) to (e) for ${\bf L_P}^{\sX,\infty,\sT}({\bf C})$. Let $M^\xi=({\bf F^\xi},{\bf G^{1,\xi}},{\bf G^{2,\xi}})$. Then $M^\xi$ satisfies the feasibility conditions (a) to (e) for ${\bf L_P}^{\sX^{\xi}, \sT}({\bf C^\xi})$.
\end{lemma}
\begin{proof} {\small
We want to show that the family $({\bf F^\xi},{\bf G^{1,\xi}},{\bf G^{2,\xi}})$ satisfy the feasibility conditions (a) to (e) in Linear Program~\ref{LP:primal} (where now $\sX = \{ x_0 , \ldots , x_J, \xi \}$, the sums range over $0 \leq i,j,k \leq J+1$ and the probabilities are given by the matrix ${\bf P}^\xi = ({p}^\xi_{j,n})_{0 \leq j \leq J+1, 1 \leq n \leq N}$).

(a)
Suppose $1 \leq n \leq N-1$ and $\delta \in \{1,2\}$.
For $0 \leq j \leq J-1$,
\[ 
\sum_{0 \leq k \leq J+1}  g^{\delta,\xi}_{j,k,n}  =  \sum_{0 \leq k \leq J-1} g^{\delta}_{j,k,n} + \left( g^{\delta}_{j,J,n} - \frac{g^{\delta}_{j,J+1,n}}{\xi - x_J} \right) + \frac{g^{\delta}_{j,J+1,n}}{\xi - x_J}
 =  \sum_{0 \leq k \leq J} g^{\delta}_{j,k,n}
\] 
and so
\[ \sum_{0 \leq k \leq J+1} (g^{1,\xi}_{j,k,n} + g^{2,\xi}_{j,k,n}) = \sum_{0 \leq k \leq J} (g^{1}_{j,k,n} + g^{2}_{j,k,n}) = \hat{p}_{j,n} = p^\xi_{j,n} . \]
For $j=J$,
\begin{eqnarray*}
\sum_{0 \leq k \leq J+1}  g^{\delta,\xi}_{J,k,n}  &= & \sum_{0 \leq k \leq J-1} g^{\delta}_{J,k,n} +
g^{\delta}_{J,J,n} - \frac{g^\delta_{J,J+1,n}}{\xi - x_J} - \frac{g^\delta_{J+1,J+1,n}}{\xi - x_J} + \frac{g^{\delta}_{J,J+1,n}}{\xi - x_J} \\
& = & \sum_{0 \leq k \leq J} g^{\delta}_{J,k,n} - \frac{g^\delta_{J+1,J+1,n}}{\xi - x_J} \\
& = & \sum_{0 \leq k \leq J} g^{\delta}_{J,k,n} - \sum_{0 \leq k \leq J+1}\frac{g^\delta_{J+1,k,n}}{\xi - x_J}.
\end{eqnarray*}
Then
\[ \sum_{0 \leq k \leq J+1} (g^{1,\xi}_{J,k,n} + g^{2,\xi}_{J,k,n}) = \hat{p}_{J,n} - \frac{\hat{p}_{J+1,n}}{\xi - x_J} = p^\xi_{j,n} . \]
For $j = J+1$,
\[
\sum_{0 \leq k \leq J+1}  g^{\delta,\xi}_{J+1,k,n}  = \sum_{0 \leq k \leq J+1} \frac{g^{\delta}_{J+1,k,n}}{\xi - x_J} = \frac{\hat{p}_{J+1,n}}{\xi - x_J} = p^\xi_{J+1,n} . \]

(b)
Suppose $2 \leq n \leq N$
For $0 \leq j \leq J-1$,
\[
\sum_{0 \leq i \leq J+1}  g^{\delta,\xi}_{i,j,n-1}  =  \sum_{0 \leq i \leq J} g^{\delta}_{i,j,n-1} + \frac{g^{\delta}_{J+1,j,n-1}}{\xi - x_J} = \sum_{0 \leq i \leq J} g^{\delta}_{i,j,n-1} \]
since $g^{\delta}_{J+1,j,n-1}=0$. Then
\[ \sum_{0 \leq i \leq J+1}  (g^{1,\xi}_{i,j,n-1} + g^{2,\xi}_{i,j,n-1}) =  \sum_{0 \leq i \leq J} (g^{1}_{i,j,n-1}+g^{2}_{i,j,n-1}) = \hat{p}_{j,n} = p^\xi_{j,n}. \]
For $j=J$,
\begin{eqnarray*}
\sum_{0 \leq i \leq J+1}  g^{\delta,\xi}_{i,J,n-1} & = &  \sum_{0 \leq i \leq J-1} \left( g^{\delta}_{i,J,n-1} - \frac{g^{\delta}_{i,J+1,n-1}}{\xi - x_J} \right) \\
&& \hspace{10mm} + g^{\delta}_{J,J,n-1} - \frac{g^\delta_{J,J+1,n-1}}{\xi - x_J} - \frac{g^\delta_{J+1,J+1,n-1}}{\xi - x_J}
+ \frac{g^{\delta}_{J+1,J,n-1}}{\xi - x_J}         \\
& = & \sum_{0 \leq i \leq J} g^{\delta}_{i,J,n-1} - \sum_{0 \leq i \leq J+1} \frac{g^{\delta}_{i,J+1,n-1}}{\xi -x_J}
\end{eqnarray*}
Then $\sum_{0 \leq i \leq J+1}  (g^{1,\xi}_{i,J,n-1} + g^{2,\xi}_{i,J,n-1}) =  \hat{p}_{J,n} - \frac{\hat{p}_{J+1,n}}{\xi - x_J} = p^\xi_{J,n}$.

For $j=J+1$
\[
\sum_{0 \leq i \leq J+1}  g^{\delta,\xi}_{i,J+1,n-1} = \sum_{0 \leq i \leq J+1}  \frac{g^{\delta}_{i,J+1,n-1}}{\xi - x_J}  \]
and so
\[ \sum_{0 \leq i \leq J+1}  (g^{1,\xi}_{i,J+1,n-1} + g^{2,\xi}_{i,J+1,n-1}) =  \frac{1}{\xi - x_J} \sum_{0 \leq i \leq J+1} (g^{1}_{i,J+1,n-1} + g^{2}_{i,J+1,n-1})
= \frac{\hat{p}_{J+1,n}}{\xi - x_J} = p^\xi_{J+1,n}. \]

(c) and (d)
Now for the martingale condition. Fix $1 \leq n \leq N-1$ and $\delta \in \{1,2\}$.  For $ 0 \leq j \leq J-1$
\begin{eqnarray*}
 \sum_{0 \leq k \leq J+1} (x_k - x_j) g^{\delta,\xi}_{j,k,n} & = & \sum_{0 \leq k \leq J-1} (x_k - x_j) g^{\delta}_{j,k,n} + (x_J - x_j) \left( g^{\delta}_{j,J,n}
- \frac{g^{\delta}_{j,J+1,n}}{\xi - x_J} \right) \\
&&  \hspace{10mm} +  (\xi - x_j) \frac{g^{\delta}_{j,J+1,n}}{\xi - x_J} \\
 & = & \sum_{0 \leq k \leq J} (x_k - x_j) g^{\delta}_{j,k,n} + g^{\delta}_{j,J+1,n} = 0.
\end{eqnarray*}
For $j=J$,
\begin{eqnarray*}
 \sum_{0 \leq k \leq J+1} (x_k - x_J) g^{\delta,\xi}_{J,k,n} & = & \sum_{0 \leq k \leq J-1} (x_k - x_J) g^{\delta}_{j,k,n} + (\xi - x_J) \frac{g^{\delta}_{J,J+1,n}}{\xi - x_J} \\
 & = & \sum_{0 \leq k \leq J} (x_k - x_J) g^{\delta}_{j,k,n} + g^{\delta}_{J,J+1,n} = 0,
\end{eqnarray*}
and for $j=J+1$
\[ \sum_{0 \leq k \leq J+1} (x_k - \xi) g^{\delta,\xi}_{J+1,k,n} =\sum_{0 \leq k \leq J} (x_k - \xi) g^{\delta}_{J+1,k,n} = 0 \]
since each of the terms in the sum is zero.

\item[(e)] For $0 \leq j \leq J-1$,
\begin{eqnarray*}
f^\xi_{j,1} - \sum_{0 \leq k \leq J+1} g^{2,\xi}_{j,k,1} & = & f_{j,1} - \sum_{0 \leq k \leq J-1} g^{2}_{j,k,1} - \left(g^2_{j,J,1} - \frac{g^2_{j,J+1,1}}{\xi - x_J} \right) - \frac{g^2_{j,J+1,1}}{\xi - x_J} \\
& = & f_{j,1} - \sum_{0 \leq k \leq J} g^{2}_{j,k,1} \leq 0.
\end{eqnarray*}

For $j=J$
\begin{eqnarray*}
f^\xi_{J,1} - \sum_{0 \leq k \leq J+1} g^{2,\xi}_{J,k,1} & = & f_{j,1} - \frac{g^2_{J+1,J+1,1}}{(\xi - x_J)}  - \sum_{0 \leq k \leq J-1} g^2_{J,k,1} \\
&& \hspace{10mm}  - g^2_{J,J,1} + \frac{1}{(\xi - x_J)} (g^2_{J,J+1,1} + g^2_{J+1,J+1,1}) - \frac{g^2_{J,J+1,1}}{\xi - x_J} \\
& = & f_{J,1} - \sum_{0 \leq k \leq J} g^{2}_{J,k,1}  \leq 0.
\end{eqnarray*}
For $j = J+1$
\[
f^\xi_{J+1,1} - \sum_{0 \leq k \leq J+1} g^{2,\xi}_{J+1,k,1}  = \frac{f_{J+1,1}}{\xi-x_J} - \frac{\sum_{0 \leq k \leq J+1} g^2_{J+1,k,1}}{(\xi - x_J)} \leq 0.
\]

For $2 \leq n \leq N-1$, and $0 \leq j \leq J-1$,
\begin{eqnarray*}
\lefteqn{f^\xi_{j,n} - \sum_{0 \leq k \leq J+1} g^{2,\xi}_{j,k,n} + \sum_{0 \leq i \leq J+1} g^{2,\xi}_{i,j,n-1}} \\
& = & f_{j,n} - \left[ \sum_{0 \leq k \leq J-1} g^{2}_{j,k,n} + \left(g^2_{j,J,n} - \frac{g^2_{j,J+1,n}}{\xi - x_J} \right) + \frac{g^2_{j,J+1,n}}{\xi - x_J} \right] + \sum_{0 \leq i \leq J} g^2_{i,j,n-1} \\
& = & f_{j,n} - \sum_{0 \leq k \leq J} g^{2}_{j,k,n} + \sum_{0 \leq i \leq J} g^{2}_{i,j,n-1} \leq 0 .
\end{eqnarray*}
For $j=J$,
\begin{eqnarray*}
\lefteqn{f^\xi_{J,n} - \sum_{0 \leq k \leq J+1} g^{2,\xi}_{J,k,n} + \sum_{0 \leq i \leq J+1} g^{2,\xi}_{i,J,n-1}} \\
& = & f_{J,n} - \frac{1}{(\xi - x_J)} (g^2_{J+1,J+1,n} - g^2_{J,J+1,n-1} - g^2_{J+1,J+1,n-1}) \\
&&   - \sum_{0 \leq k \leq J-1} g^2_{J,k,n} - g^2_{J,J,n} + \frac{1}{(\xi - x_J)} (g^2_{J,J+1,n} + g^2_{J+1,J+1,n}) - \frac{g^2_{J,J+1,n}}{\xi - x_J} \\
&& + \sum_{0 \leq i \leq J-1} g^2_{i,J,n-1} + g^2_{J,J,n-1} - \frac{g^2_{J,J+1,n-1} + g^2_{J+1,J+1,n-1}}{\xi - x_J}  \\
& = & f_{J,n} - \sum_{0 \leq k \leq J} g^{2}_{J,k,n} + \sum_{0 \leq i \leq J} g^{2}_{i,J,n-1} \leq 0 .
\end{eqnarray*}
For $j=J+1$,
\begin{eqnarray*}
\lefteqn{ f^\xi_{J+1,n} - \sum_{0 \leq k \leq J+1} g^{2,\xi}_{J+1,k,n} + \sum_{0 \leq i \leq J+1} g^{2,\xi}_{i,J+1,n-1} } \\
& = & \frac{1}{\xi - x_J} \left[ f_{J+1,n} - \sum_{0 \leq k \leq J+1} g^{2}_{J+1,k,n} + \sum_{0 \leq i \leq J+1} g^{2}_{i,J+1,n-1} \right] \leq 0.
\end{eqnarray*}

For $n=N$, and $0 \leq j \leq J-1$,
\[ f^\xi_{j,N} + \sum_{0 \leq i \leq J+1} g^{2,\xi}_{i,j,N-1}  =  f_{j,N} + \sum_{0 \leq i \leq J} g^2_{i,j,N-1}
 \leq \hat{p}_{j,N} = p^\xi_{j,N}.
\]

For $j=J$,
\begin{eqnarray*}
\lefteqn{f^\xi_{J,N}  + \sum_{0 \leq i \leq J+1} g^{2,\xi}_{i,J,N-1} } \\
& = & f_{J,N} + \frac{1}{\xi - x_J} \left(g^2_{J,J+1,N-1} + g^2_{J+1,J+1,N-1} - c_{J,N} \right) \\
&& \hspace{10mm} + \sum_{0 \leq i \leq J-1} g^2_{i,J,N-1} +  g^2_{J,J,N-1} - \frac{1}{\xi - x_J} \left(g^2_{J,J+1,N-1} + g^2_{J+1,J+1,N-1}  \right)\\
& = & f_{J,N} + \sum_{0 \leq i \leq J} g^{2}_{i,J,N-1} - \frac{c_{J,N}}{\xi - x_J} \leq \hat{p}_{J,N} - \frac{c_{J,N}}{\xi - x_J} = {p}^{\xi}_{J,N}.
\end{eqnarray*}
For $j=J+1$,
\begin{eqnarray*}
 f^\xi_{J+1,N} + \sum_{0 \leq i \leq J+1} g^{2,\xi}_{i,J+1,N-1}
& = & \frac{f_{J+1,N}}{(\xi - x_J)}  + \sum_{0 \leq i \leq J+1} \frac{g^2_{i,J+1,N-1}}{\xi - x_J} \\
& \leq & \frac{\hat{p}_{J+1,N}}{\xi - x_J} = p^\xi_{J+1,N} .
\end{eqnarray*}

} 
\end{proof}

\begin{corollary}
\label{cor:mprice}
Suppose $M=({\bf F},{\bf G^1},{\bf G^2})$ is feasible for ${\bf L}_{\bf P}^{\sX, \infty, \sT}({\bf C})$ and define $\phi_*(M)$ via (\ref{eq:phiM}). Set $M^\xi=({\bf F^\xi},{\bf G^{1,\xi}},{\bf G^{2,\xi}})$ and set $\phi_*(M^\xi) = \sum_{0 \leq j \leq J+1} a(x_j,t_n) f^\xi_{j,n}$ where $x_{J+1} = \xi$.
If $\Upsilon = \phi_*(M^\xi) - \phi_*(M)$ then
\begin{eqnarray*}
\Upsilon & = & \sum_{1 \leq n \leq N} \left[ f^\xi_{J+1,n} a(\xi,t_n) - f_{J+1,n} \lim_{x \uparrow \infty} \frac{a(x,t_n)}{x} \right] + \sum_{1 \leq n \leq N} a(x_J,t_n) \left( f^{\xi}_{J,n} - f_{J,n} \right) \\
 & = & \sum_{1 \leq n \leq N} f_{J+1,n} \left[ \frac{a(\xi,t_n)}{\xi - x_J} - \lim \frac{a(x,t_n)}{x} \right] - \frac{a(x_J,t_N)c_{J,N}}{\xi - x_J} \\
 && - \frac{1}{\xi - x_J} \sum_{1 \leq n \leq N-1} [a(x_J,t_n) - a(x_J, t_{n+1})]g^2_{J+1,J+1,n} \\
 &&    + \frac{1}{\xi - x_J} \sum_{2 \leq n \leq N} a(x_J,t_n)g^2_{J,J+1,n-1}
\end{eqnarray*}
In particular $\phi_*(M^\xi) - \phi_*(M) \geq - \frac{\Upsilon_0}{\xi - x_J}$ for some constant $\Upsilon_0$ independent of $\xi$.
\end{corollary}

\begin{proof}
The calculation of $\Upsilon$ is straightforward. For the final statement note that if $R_n = \lim_{x \uparrow \infty} \frac{a(x,t_n)}{x}$ then since $a(\cdot, t_n)$ is convex we have $a(x,t_n) \geq \alpha_n + R_n x$ for some $\alpha_n$. Then $a(\xi, t_n) - (\xi - x_J) R_n \geq \alpha_n + R_n x_J  \geq \alpha_n$. Then we can take
\[ \Upsilon_0 = a(x_J,t_N)c_{J,N} + \sum_{1 \leq n \leq N-1} [a(x_J,t_n) - a(x_J, t_{n+1})]g^2_{J+1,J+1,n} - \sum_{1 \leq n \leq N} f_{J+1,n} \alpha_n. \]
\end{proof}

Given a feasible triple $({\bf F}, {\bf G^1}, {\bf G^2})$ we cannot identify it directly with a model. However, if we take $\xi > \xi_0$ then we can hope to construct candidate models in $\sM_2^{\sX^\xi,\sT}({\bf C})$. But, it may not be the case that the matrices $({\bf F^\xi}, {\bf G^{1,\xi}}, {\bf G^{2,\xi}})$ are non-negative, so the candidate model $M^\xi$ may not be feasible. To circumvent this issue we mix such candidate models with other models in $\sM_2^{\R^+, \sT}({\bf C})$ for which the entries
are non-negative. We show that by varying this mixture we can find a consistent model for which the model price of the American option is arbitrarily close to the super-replication price $\Psi^{\sX,\infty,\sT}$.

We begin with a useful lemma.
\begin{lemma}
\label{lem:jtlaw}
Let $\nu_1$ and $\nu_2$ be probability measures on a discrete set $\sY = \{y_1,  \ldots , y_M\}$ where $y_1 < \ldots < y_M$. Suppose that $\nu_1$ is less than or equal to $\nu_2$ in convex order. Then there exists a joint law $\rho$ on $\sY \times \sY$
such that the $i^{th}$ marginal of $\rho$ is $\nu_i$ and $\sum_{k} (y_k - y_j) \rho( \{ (y_j,y_k) \} ) = 0$ for all $j$.

Suppose further that $\E^{\nu_2}[(Y - y_m)^+] > \E^{\nu_1}[(Y-y_m)^+]$ for all $2 \leq m \leq M-1$ and that $\nu_i(\{ y_m \})>0$ for $i=1,2$ and all $1 \leq m \leq M$.
Then, the joint law can be chosen so that $\rho( \{ (y_j,y_k) \} ) > 0$ for all $2 \leq j \leq M-1$ and $1 \leq k \leq M$.
\end{lemma}

\begin{proof}
The existence result in the first paragraph is classical and follows from results of Strassen~\cite{Strassen:65}. The existence result in the second paragraph follows for suitable interpretation of a well chosen solution for the Skorokhod embedding problem for a non-trivial initial law, see Hobson~\cite{Hobson:11}.
A solution based on the stopping of a skip-free martingale Markov chain on $\sY$ at an independent exponential time suffices, see Cox et al~\cite{CoxHobsonObloj:08}.
\end{proof}

\begin{lemma}
\label{lem:converges}
\[ \sup_{M \in \sM^{\R^+, \sT}({\bf C})} \phi(M) \geq \Psi^{\sX,\infty,\sT}(a,{\bf C}). \]
\end{lemma}

\begin{proof}
Given $\epsilon>0$ we aim to show how to choose a consistent model $M$ such that $\phi_*(M) > \Psi^{\sX,\infty,\sT}(a,{\bf C}) - \epsilon$. Abbreviate $\Psi^{\sX,\infty,\sT}(a,{\bf C})$ to $\Psi$.

Let $M = ({\bf F}, {\bf G^1}, {\bf G^2})$ be the optimiser in Linear Program~\ref{LP:primalA}, so that $\phi_*(M) = \Psi$.

For $\xi > \xi_0$, let $M^\xi = ({\bf F^{\xi}},{\bf G^{1,\xi}}, {\bf G^{2,\xi}})$ be the triple of matrices defined just before Lemma~\ref{lem:verylong}.
Set
\( \xi_1 = \max \left\{ \xi_0, x_J + \frac{2}{\epsilon} \Upsilon_0 \right\} \).
Then, by Corollary~\ref{cor:mprice} for $\xi > \xi_1$, $\phi_*(M^{\xi}) \geq \Psi - \epsilon/2$.

If all the elements $g^{\delta}_{j,J,n}$ (with $j \leq J$) are strictly positive then for large enough $\xi > \xi_1$, $M^{\xi} = ({\bf F^{\xi}},{\bf G^{1,\xi}}, {\bf G^{1,\xi}}) \geq 0$, $M^\xi \in \sM^{\sX^\xi, \sT}({\bf C^\xi}) \subseteq \sM^{\R^+,\sT}({\bf C})$ is a feasible model and we are done. More generally we may have $g^{\delta}_{j,J,n} = 0$ for some $0 \leq j \leq J$ and then $M^{\xi}$ is not feasible for any $\xi$.

Fix $\tilde{\xi} > \xi_0$. Let $\tilde{\sX} = X^{\tilde{\xi}} = \sX \cup \{ \tilde{\xi} \}$ and let the $(J+2) \times N$ matrix ${\bf \tilde{C}}$ be given by
$\tilde{c}_{j,n} = c_{j,n}$ for $0 \leq j \leq J$ and $\tilde{c}_{J+1,n} = 0$.
Accordingly we can define the matrix ${\bf \tilde{P}}$ by, for $1 \leq n \leq N$,
\[ \tilde{p}_{j,n} = p_{j,n} (1 \leq j \leq J-1) \hspace{10mm} \tilde{p}_{J,n} = p_{J,n} - \frac{c_{J,n}}{\xi - x_J} \hspace{10mm} \tilde{p}_{J+1,n} = \frac{c_{J,n}}{\xi - x_J}. \]

Let $\tilde{\nu}_n$ denote the law on $\tilde{\sX}$ such that $\tilde{\nu}_n(\{ x_j \}) = \tilde{p}_{j,n}$. Then, by Proposition~\ref{prop:markov} $\sM^{\tilde{\sX},\sT}_1({\bf \tilde{C}})$ is non-empty
and by Lemma~\ref{lem:jtlaw}, there exists a model $\tilde{M}_1 \in \sM^{\tilde{\sX},\sT}_1({\bf \tilde{C}})$ such that the probability of every transition (except those away from the absorbing endpoints) is positive: ie. $\eta > 0$ where
\[ \eta = \min_{1 \leq n \leq N-1} \min_{1 \leq i \leq J} \min_{0 \leq j \leq J+1} \Prob^{\tilde{M}_1}(X_{t_n} = i, X_{t_{n+1}} = j). \]
Set $\tilde{g}_{i,j,n} = \Prob^{\tilde{M}_1}(X_{t_n} = i, X_{t_{n+1}} = j)$.

Define  $\tilde{M}_2 \in \sM^{\tilde{\sX},\sT}_2({\bf \tilde{C}})$ by
\[ \tilde{g}^1_{i,j,n} = \tilde{g}_{i,j,n} \frac{N-n}{N} \hspace{10mm} \tilde{g}^2_{i,j,n} = \tilde{g}_{i,j,n} \frac{n}{N} \hspace{10mm} \tilde{f}_{j,n} = \frac{\tilde{p}_{j,n}}{N} \]
and note that $\tilde{g}^\delta_{i,j,n} \geq \eta/N$.
$\tilde{M}_2$ is obtained from $\tilde{M}_1$ by augmenting the price process with a regime process which jumps to state 2 at a time which is uniformly distributed on $\{1, \ldots ,N\}$ and is independent of the price process.

Choose $\zeta < \frac{\epsilon}{2\Psi}$ and $\xi_2 > \max\{\tilde{\xi}, \xi_1,  x_J + \frac{N(1 - \zeta)}{\zeta \eta} \}$.
Let $\hat{\sX} = \sX^{\tilde{\xi}, \xi_2} = \sX \cup \{ \tilde{\xi}, \xi_2 \}$. We construct a model $\hat{M} \in \sM^{\hat{\sX},\sT}_2$ which is consistent with ${\bf C}$ on $\sX \times \sT$ and is therefore an element of $\sM^{\R^+,\sT}({\bf C})$. Set $x_{J+1} = \tilde{\xi}$ and $x_{J+2} = \xi_2$.
For $\delta = 1,2$, $1 \leq n \leq N$ and $0 \leq j,k \leq J$ define
\[ \hat{g}^\delta_{j,k,n} = \zeta \tilde{g}^\delta_{j,k,n} + (1 - \zeta) g^{\delta,\xi_2}_{j,k,n}  \]
and set also
\[ \hat{g}^\delta_{j,J+1,n} = \zeta \tilde{g}^\delta_{j,J+1,n}
\hspace{15mm} \hat{g}^\delta_{j,J+2,n} = (1 - \zeta) g^{\delta,\xi_2}_{j,J+1,n} \]
\[ \hat{g}^\delta_{J+1,J+1,n} = \zeta \tilde{g}^\delta_{J+1,J+1,n}
\hspace{15mm} \hat{g}^\delta_{J+2,J+2,n} =  (1 - \zeta) g^{\delta,\xi_2}_{J+1,J+1,n} \]
together with $\hat{g}^\delta_{J+1,J+2,n} = 0 = \hat{g}^\delta_{J+2,J+1,n}$ and $\hat{g}^\delta_{J+1,j,n} = 0 = \hat{g}^\delta_{J+2,j,n}$ for all $0 \leq j \leq J$.

It follows that since $g^{\delta,\xi_2}_{j,J,n} \geq - \frac{1}{\xi_2 - x_J}$ we have for $0 \leq j \leq J$
\[ \hat{g}^\delta_{j,J,n} = \zeta \left( \tilde{g}^\delta_{j,J,n} + \frac{(1 - \zeta)}{\zeta} g^{\delta,\xi_2}_{j,J,n} \right) \geq \zeta \left( \frac{\eta}{N} -  \frac{(1 - \zeta)}{\zeta} \frac{1}{\xi_2 - x_J} \right) \geq 0, \]
and the probabilities $\hat{g}^\delta_{j,k,n}$ define a model $\hat{M}$. Moreover, the model is a mixture of the two models $\tilde{M}_2$ and $M^{\xi_2}$ which individually are consistent with ${\bf C}$ on $\sX \times \sT$, and hence $\hat{M}$ is consistent with ${\bf C}$ and $\hat{M} \in \sM^{\R^+,\sT}({\bf C})$.
Finally,
\[ \phi_*(\hat{M}) \geq (1 - \zeta) \phi_*( M^{\xi_2} ) > \left( 1 - \frac{\epsilon}{2 \Psi} \right) \left( \Psi - \frac{\epsilon}{2} \right) > \Psi - \epsilon. \]

\end{proof}

Define ${\sP}^{\R^+, \sT}(a, {\bf C}) = \sup_{M \in \sM^{\R^+, \sT}({\bf C}) } \sup_{\tau \in \sT} \E^M[a(X_\tau,\tau)]$ and
${\sH}^{\R^+, \sT}(a, {\bf C}) = \inf_{({\bf B},\Theta) \in \sS^{\R^+,\sT}(a)} H_{\bf C}({\bf B})$.

\begin{theorem}
\label{thm:main2} We have $\Phi^{\sX,\infty,\sT}(a, {\bf C}) = {\sP}^{\R^+, \sT}(a, {\bf C}) = {\sH}^{\R^+, \sT}(a, {\bf C}) = \Psi^{\sX,\infty,\sT}(a, {\bf C})$. In particular,
there is a sequence of elements of $\sM^{\R^+,\sT}({\bf C})$ for which the model based price converges to $\Psi^{\sX,\infty,\sT}$, and
there is a super-replicating strategy of the form described in Definition~\ref{def:tradingstrategy} for which the cost of the strategy is the lowest amongst the class of all semi-static super-replicating strategies.
\end{theorem}

\begin{proof}
We have ${\sP}^{\R^+, \sT}(a, {\bf C}) \leq {\sH}^{\R^+, \sT}(a, {\bf C}) \leq \Psi^{\sX,\infty,\sT} = \Phi^{\sX,\infty,\sT}$ by weak duality, the fact that the optimiser in Linear Program~\ref{LP:dualinfinityA} is a super-replicating semi-static strategy, and the duality between ${\bf L}_{\bf H}^{\sX, \infty, \sT}$ and ${\bf L}_{\bf P}^{\sX, \infty, \sT}$. But Lemma~\ref{lem:converges} gives ${\sP}^{\R^+, \sT}(a, {\bf C}) \geq \Psi^{\sX,\infty,\sT}(a, {\bf C})$. Hence there is equality throughout.
\end{proof}

\section{Further Examples: The American put, choice of filtration and a continuum of calls}

\subsection{A numerical example: the American put}
\label{ssec:numerical}

This section contains numerical examples involving the American put. The section has multiple aims: to illustrate the theory, to demonstrate the tightness of the bound in practice and to show how much of the early exercise premium can be attributed to the various features.

We assume throughout a constant risk-free interest rate of 5\%, a fixed initial value (100) of the underlying, and largest maturity $T=1$. In addition we assume that the prices of traded European options are consistent with a constant-volatility Black-Scholes model with annualized volatility of 20\%. The discounted payoff of the American option is $e^{-rt}a_S(S_t,t) = e^{-rt}(K-S_t) = (e^{-rt} K- X_t)^+ = a(X_t,t)$ where $X_t = e^{-rt}S_t$. Our goal is to price American options on $X$ with payoff $a(X_t,t)$. Note that $X$ is a martingale under any consistent pricing measure.

Let $\chi(A)$ denote the price of the American option with payoff $A$ on $X$ under the Black-Scholes model (with volatility 20\%). Let $\phi(A)=\phi(A,{\bf C})$ denote the highest model-based price of the same American option $A$ where models are consistent with the matirx of call prices ${\bf C}$. For fixed $t \in [0,T]$ let $\zeta(t,A) = \sup_{M \in \sM^{\R^+,\bbT}({\bf C})} \E^M[A(X_t,t)]$ denote the corresponding European price with maturity $t$ (note that the only relevant feature of the model is the law of $X$ at time $t$) and let $\zeta(A) = \max_{0 \leq t \leq T} \zeta(t,A)$. $\chi(A)$ and $\zeta(A)$ provide two benchmarks against which to compare the model-free bound $\phi(A)$. The difference $\chi(A) - \zeta(A)$ can be taken as the size of the American premium using conventional modelling assumptions. The aim is to compare this quantity with $\phi(A)-\zeta(A)$ which is the size of the largest American premium when we search over all models which are consistent with the traded options data.

Note we use $\zeta(A)$ rather than $\zeta(T,A)$ as the benchmark for the European option because in some circumstances (high interest rates and high strikes) it is optimal to exercise an in-the-money American put instantly under virtually any model. Using $\zeta(T,A)$ as the benchmark would suggest that the value of the American feature of the option was large, although in almost all cases the optimal strategy would be to exercise immediately.

We assume that the set of traded options have (strike, maturity) pairs in $\sK \times \sT$. When we calculate the model-free bound on the American put with payoff $a$ we could equivalently use the (piecewise linear in strike, piecewise constant in maturity) payoff $\overline{a}=\overline{a}^{\sX,\sT}$ where for $t_n \leq t < t_{n+1}, x_j \leq x < x_{j+1}$
\[ \overline{a}(x,t) = \frac{x_{j+1}-x}{x_{j+1} - x_j} a(x_j ,t_n) + \frac{x-x_j}{x_{j+1} - x_j}a( x_{j+1}, t_{n}) . \]
In particular, $\phi(a) = \phi(\overline{a})$.
However, in calculating $\chi$ this change of payoff function does affect the value and typically we have $\chi(\overline{a})> \chi(a)$. For example, consider the American put with strike 100, so that $a(x,t)=(e^{-rt}100 - x)^+$. Suppose $\sT = \{1/4, 1/2, 3/4, 1\}$ and $\sK = \{70,80,90,100,110,120,130,140\}$. Then $\chi(a)=6.09$ whereas $\chi(\overline{a}) = 6.74$. In comparison $\phi(a) = \phi(\overline{a}) = 7.66$. Thus, of the difference between the model-independent premium and the Black-Scholes price $\phi(a) - \chi(a)=1.57$, $6.74-6.09 =0.65$ is attributed to the effect of the mesh, and 0.92 is attributed to the modelling assumptions of the Black-Scholes model.

We also find $\zeta(\overline{a})= 6.35$. Then, under the Black-Scholes model the value of the American feature for the claim $\overline{a}$ is $6.74-6.35 = 0.38$, whereas a supremum over all models yields $\phi(\overline{a}) - \zeta(\overline{a}) = 7.66 - 6.35 = 1.31$. Thus the Black-Scholes estimate of the American premium is just 29\% of the possible largest American premium.

Figure~\ref{tab:1} shows the effect of changing the mesh size.
The first four columns describe the family of European options which are traded. The maturities are evenly spaced at $\{ 1/N, 2/N, \ldots,1\}$. The headings `Lowest' and `Highest' refer to the lowest and highest strikes, and `Interval' refers to the interval between strikes, so that except for the first and fourth rows $\sK = \{70,80,90,100,110,120,130,140\}$. The next three columns give various prices of the linearized option: the supremum over models of American prices, the Black-Scholes American price and the European price. The final column gives the ratio $[\chi(\overline{a})- \zeta(\overline{a})]/[\phi(\overline{a})- \zeta(\overline{a})]$ expressed as a percentage.

As the grid becomes finer the values of $\zeta(\overline{a})$, $\phi(\overline{a})$ and $\chi(\overline{a})$ all fall. However, the proportion of the American premium which is captured by a Black-Scholes valuation remains broadly constant at roughly one quarter to one third.

\begin{figure}
\centering
\begin{tabular}{||cccc|ccc|c||}
\hline
Maturities       & Lowest & Highest  & Interval  & $\phi(\overline{a})$ & $\chi(\overline{a})$ & $\zeta(\overline{a})$ & \% Premuium \\
\hline
2                & 100           & 100             & -         &  7.91                & 7.80                 & 7.77                  & 22.4  \\
2                & 70            & 140             & 10        &  7.79                & 7.20                 & 6.89                  & 35.3     \\
4                & 70            & 140             & 10        &  7.66                & 6.74                 & 6.35                  & 29.1   \\
4                & 70            & 140             & 2.5       &  7.65                & 6.57                 & 6.13                  & 29.0           \\
12               & 70            & 140             & 10        &  7.55                & 6.42                 & 6.00                  & 27.0       \\
26               & 70            & 140             & 10        &  7.51                & 6.34                 & 5.91                  & 26.8     \\
\hline
\end{tabular}
\caption{The table shows the value of an at-the-money put with strike 100 and maturity 1. European options are assumed to trade on an implied volatility of 20\%. The rows in the table correspond to different grids, with the grid becoming finer and prices lower as we move down the table.}
\label{tab:1}
\end{figure}

In the next table, Figure~\ref{tab:2} we fix on a mesh, and consider the impact of varying the moneyness of the option. We use a quarterly mesh and strikes every 10 units. The main conclusion is that the failure of the Black-Scholes model to capture the full value of the American premium is most pronounced for out of the money puts, and that always the  Black-Scholes model captures less than half the maximum possible value of the American feature.

\begin{figure}
\centering
\begin{tabular}{||c|ccc|c||}
\hline
Strike       & $\phi(\overline{a})$ & $\chi(\overline{a})$ & $\zeta(\overline{a})$ & \% Premuium \\
\hline
 80         &  1.00                & 0.92                 & 0.91                  & 15.5  \\
 90         &  3.25                & 2.89                 & 2.79                  & 20.6    \\
 100        &  7.66                & 6.74                 & 6.35                  & 29.1   \\
 110        &  14.09               & 12.81                & 11.73                 & 45.8           \\
120         &  22.02               & 20.89                & 20.15                 & 39.6       \\
\hline
\end{tabular}
\caption{The effect of moneyness on option value. Maturities are $\{1/4,1/2,3/4,1\}$ and $\sK = \{ 70,80,90, \ldots, 140 \}$.}
\label{tab:2}
\end{figure}

The conclusion from this section is that pricing under the Black-Scholes model can greatly undervalue the American feature of the option and the ability of the option holder to respond to resolution of  model uncertainty.

\subsection{The choice of filtration}
\label{ssec:filtration}

The purpose of this section is to illustrate how the choice of filtration can have a large impact on the range of possible prices of the American option.
Restricting attention to models based on the natural filtration yields underestimates of the value of the American option.

In this section we consider the following very simple example. Time ranges over $\sT_0 = \{0,1,2 \}$. At $t=0$ we have $X_0=2$, at $t=1$ $X$ takes values in $\{1,3\}$, at $t=2$, $X$ takes values in $\{0,2,4\}$. The martingale property, together with the fact that the state-space is so simple means that if we are given the price of one Arrow-Debreu security at time $2$ then the marginal laws of $X$ are fully specified. We suppose that there is a security which for price $2/5$ pays 1 in state $(2,4)$. Then under any consistent model $X_1$ has uniform law on $\{ 1,3 \}$, and $X_2$ has law $\{2/5, 1/5, 2/5 \}$ on $\{0,2,4\}$.

We want to value the American option which pays $1$ in state $(1,1)$ and 8 in state $(2,4)$ and otherwise pays 0. See Figure~\ref{fig:simpleeg}.

\begin{figure}[!htbp]
\centering
\begin{tikzpicture}
 \draw[dashed] (0,0) -- (2,1) ;
 \draw[dashed] (0,0) -- (2,-1);
 \draw[dashed] (2,1) -- (4,0) ;
 \draw[dashed] (2,1) -- (4,-2);
 \draw[dashed] (2,1) -- (4,2);
 \draw[dashed] (2,-1) -- (4,-2);
 \draw[dashed] (2,-1) -- (4,0);
 \draw[dashed] (2,-1) -- (4,2);
 \draw [fill] (4,2) circle [radius=.1] node[right]{(4,2,{\cred 8})};
  \draw [fill] (4,0) circle [radius=.1] node[right]{(2,2,{\cred 0})};
  \draw [fill] (4,-2) circle [radius=.1] node[right]{(0,2,{\cred 0})};
  \draw [fill] (0,0) circle [radius=.1] node [left] {(2,0,{\cred 0})};
  \draw [fill] (2,1) circle [radius=.1] node [above left] {(3,1,{\cred 0})} ;
  \draw [fill] (2,-1) circle [radius=.1] node [below left] {(1,1,{\cred 1})} ;
\end{tikzpicture}%

\caption{{\it The space of possible paths, and the payoff of the American option}. The labels at the nodes on the graph consist of a triple, the elements of which are price level, time and payoff of the American option respectively.}
\label{fig:simpleeg}
\end{figure}
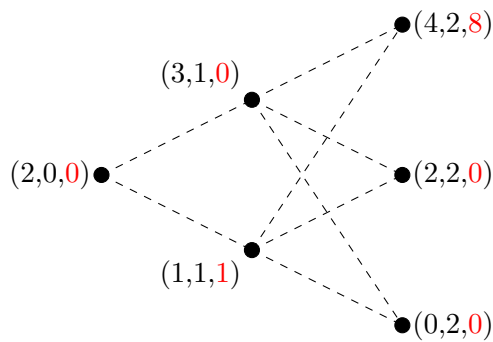

Denote by $(p,q,r)$ the transition probabilities of going from $(1,1)$ to $((4,2), (2,2), (0,2))$ respectively, and by $(s,t,u)$ the transition probabilities of going from $(3,1)$ to $((4,2), (2,2), (0,2))$. The fact that $(p,q,r)$ are martingale probabilities gives $0 \leq p \leq 1/4$ and $(q,r) = (\frac{1-4p}{2}, \frac{1}{2} + p)$. Similarly, $1/2 \leq s \leq 3/4$ and $(t,u) = (\frac{3-4s}{2}, s -\frac{1}{2})$. Finally, the constraints on the law of $X$ at $t=2$ give $p+s=4/5$. Any choice $(p,s)$ with $0 \leq p \leq 1/4$, $1/2 \leq s \leq 3/4$ and $(p+s)=4/5$ leads to a consistent model.

Suppose $X_1=1$. If the conditional transition probabilities from $(1,1)$ are $(p,q,r)$ then the value of immediate exercise at $(1,1)$ is $1$, and the value on continuation is $8p$, so that it is optimal to continue if $p \geq 1/8$. It is always optimal to continue at $(3,1)$ and the value is $8s$. The expected payoff of the American option is then $\frac{1}{2} [8(p+s) + (1 - 8p)^+] = 16/5 + (\frac{1}{2}-4p)^+$. This is maximised by taking $p$ as small as possible, ie $p = 1/20$ to give a best model based price of 7/2.

In comparison the European option with payoff $(8,0,0)$ in states $((4,2), (2,2), (0,2))$ has price 16/5 so that the American premium is 3/10.

However, this is not the highest model price in the class of consistent models. Consider a pair of models $\hat{M}$ and $\tilde{M}$. Suppose they are characterised by the sextuples $(\hat{p},\hat{q},\hat{r},\hat{s},\hat{t},\hat{u})=(1/4,0,3/4,3/4,0,1/4)$ and $(\tilde{p},\tilde{q},\tilde{r},\tilde{s},\tilde{t},\tilde{u})=(0,1/2,1/2,3/4,0,1/4)$. Note that both of these models are martingale models, and although neither model satisfies the constraint $p+s=4/5$ the mixture $M = \frac{1}{5} \hat{M} + \frac{4}{5} \tilde{M}$ does have the property that the law of $X_2$ matches the call prices. We assume that the holder of the option learns whether the world is described by $\hat{M}$ or $\tilde{M}$ at $t=1$ before he is required to decide whether to exercise the option. We will show that the price of the American option is maximised over consistent models by the model $M$.

Under $\hat{M}$ it is optimal to exercise the American option at $t=2$, and the value of the option is 4. Under $\tilde{M}$ it is optimal to exercise at $(1,1)$ and the value of the American option is $7/2$. Provided the model uncertainty is resolved by $t=1$, the price under the mixed model is 18/5 and the American premium is 2/5. In particular, the maximum of the American premium over models under an assumption that the filtration is the natural filtration of the price process is only 3/4 of the maximum of the American premium when we consider all models. (Further, if the price of the Arrow-Debreu security at $(2,4)$ is raised to 7/16, then the European price rises to 7/2, but in the model in which transition probabilities are specified at time zero by $(p,q,r,s,t,u)$ the American option price is unchanged at 7/2. The highest model based price is based on a mixture model $M = \frac{1}{2} \hat{M} + \frac{1}{2} \tilde{M}$ in which the American option price is 15/4. In the Markovian model, the American premium is zero --- it is always optimal to exercise the American option at $t=2$ --- in the mixture model the American premium is 1/4. Restricting attention to Markovian models suggests that there is no American premium, but this is not the case.)

Now we argue that 18/5 is the highest possible model price by exhibiting a semi-static super-replicating strategy with cost 18/5. Starting with 18/5, purchase 4 Arrow-Debreu securities paying 1 in state $(4,2)$ at total cost $8/5$, leaving cash of 2. In addition, hold one unit of asset over the time-period $[0,1)$, and if the American option is not exercised at $t=1$ continue to hold the unit long position until $t=2$; otherwise hold a null position in the stock over $[1,2)$.

At $t=0$, the cash holdings are 2. At $t=1$ the cash holdings are 3 in state $(3,1)$ and 1 in state $(1,1)$. This is sufficient to cover the cost of the American option if it is exercised (and note that the holdings of Arrow-Debreu securities with maturity $t=2$ are all non-negative, so there are no remaining liabilities). If the American option is not exercised at $t=1$, and if $X_1=3$ then including the payoff from the Arrow-Debreu security, at $t=2$ the strategy realises $(8,2,0)$ in the states $((4,2),(2,2),(0,2))$ respectively. If the option is not exercised at $t=1$, and if $X_1=1$ then at $t=2$ the strategy again realises $(8,2,0)$ in the states $((4,2),(2,2),(0,2))$. Hence, the given strategy is a super-replicating strategy.

The above description does not quite fall into the notation of the rest of the paper. To see how this example fits into that structure, suppose $\sX =
\{ 0,1,2,3,4 \}$ and $\sT= \{1,2 \}$. Set $p_{1,1}=\frac{1}{2}=p_{3,1}$ together with $p_{0,2}=2/5$, $p_{2,2}=1/5$ and $p_{4,2}=2/5$. Necessarily all other $p_{j,n}$ are zero. We have $a_{\cdot ,1} = (0,1,0,0,0)$ and $a_{\cdot ,2} = (0,0,0,0,8)$, although $a_{1,1}$, $a_{3,1}$, $a_{0,2}$, $a_{2,2}$ and $a_{4,2}$ are the only relevant entries.

Define $({\bf E^1},{\bf E^2},{\bf D^1},{\bf D^2}, {\bf V})$ by $e^1_{j,2}=e^2_{j,1}=0$ for $0 \leq j \leq 4$ and
\begin{eqnarray*} e^1_{\cdot ,1} = (0,1,2,3,4) && e^2_{\cdot ,2} = (0,0,0,-2,-4) \\
                  d^1_{\cdot ,1} = (1,1,1,1,1) && d^2_{\cdot ,1} = (0,0,0,2,2) \\
                  v_{\cdot ,1} = (0,1,2,5,8)   && v_{\cdot ,2} = (0,0,0,4,8)
                  \end{eqnarray*}
We have immediately that $v_{j,n} \geq a_{j,n}$ for all $n$ and $j$, so that (\ref{eq:lpi}) holds.
Also, if ${\bf \Lambda^1}$ is the matrix with entries $\lambda^1_{j,k} = (k-j)$ then $e^1_{j,1} + \lambda^1_{j,k} = e^1_{k,1} \geq - e^2_{k,2}$ and (\ref{eq:lpii}) holds. Finally, if ${\bf \Lambda^2}$ is the matrix $\lambda^2_{j,k} = 2(k-j)I_{ \{ j = 3,4 \} }$ then $v_{j,1} - e^1_{j,1} + \lambda^2_{j,k} = v_{j,2} + e^2_{j,2}  + \lambda^2_{j,k} \leq v_{k,2} + e^2_{k,2}$ so that (\ref{eq:lpiii}) holds.

Thus the strategy implicit in $({\bf E^1},{\bf E^2},{\bf D^1},{\bf D^2}, {\bf V})$ super-replicates, and at a cost $8\times \frac{2}{5} + (-4) \times \frac{2}{5} + 3 \times \frac{1}{2} + 1 \times \frac{1}{2} = \frac{18}{5}.$

\subsection{A double continuum of option prices}
\label{ssec:continuum}
The methods of the previous sections transfer to the continuous time setting, as we shall demonstrate by example. However, the continuous setting brings new challenges which we will not attempt to overcome in this paper. Instead this section is intended merely to show how the ideas might extend to this more general framework.

We suppose we are given a double continuum (in strike and maturity) of option prices and the goal is to find the model-free upper bound on the price of an American option. One issue is that in order to define the semi-static strategy it is necessary to define the gains from trade from an admissible dynamic strategy along every admissible price path. In this example we choose a very general class of admissible price paths --- it is unlikely that one could work with this wide set of paths if the option payoff was more complicated.

\begin{assumption}
\label{ass:pricepaths}
The set of possible price paths is the set $D_{x_0}([0,T])$ of right-continuous functions with left limits which start at $x_0$.
\end{assumption}

Note that some regularity on the price paths is needed to be able to define the gains from trade, but we do not assume the existence of a (pathwise) quadratic variation or any other regularity for the price path $y$ beyond $y \in D_{x_0}([0,T])$.

\begin{lemma}
\label{lem:pathineq}
Suppose $z = \{ z(t) ; 0 \leq t \leq T \} \in D_0([0,T])$. Then for all $0 \leq s \leq t \leq T$
\begin{equation}
\label{eq:pathwiseineq}
 \int_{(s, t]} dz(u) I_{ \{ z(u-) \geq 0 \} } \leq z(t)^+ - z(s)^+
\end{equation}
\end{lemma}

\begin{proof}
We can write $(s,t] \cap I_{ \{ y(u-) \geq 0 \} }$ as a union of disjoint intervals $(l_\alpha,r_\alpha]$ or $[l_\alpha,r_\alpha]$. Assume $(l_\alpha,r_\alpha]$ is one such interval with $s < l_\alpha < r_\alpha < t$.
We have
\[ \int_{(l_\alpha,r_\alpha]} dz(u) I_{ \{ z(u-) \geq 0 \} } = z(r_\alpha) - z(l_\alpha) \leq 0  \]
since $z(r_\alpha) \leq 0 \leq z(l_\alpha)$. Alternatively, if the interval is of the form $[l_\alpha,r_\alpha]$ with $l_\alpha \leq r_\alpha$ then we must have $z(l_{\alpha}-)=0$ and
$\int_{[l_\alpha,r_\alpha]} dz(u) I_{ \{ z(u-) \geq 0 \} } = z(r_\alpha) \leq 0$.

The non-zero expressions on the right-hand-side of (\ref{eq:pathwiseineq})
arise from considering the intervals straddling $s$ and $t$ if any.
\end{proof}

Consider an American-style option with payoff $a(x,t) = |x- x_0| + b(t)$ where $b$ is a decreasing differentiable function such that $b(0) - b(T) < x_0$.

Let $X$ be a discounted stock price with $X_0 = x_0$. Suppose we are given a double continuum (in strike $k$ and maturity $t$) of call prices $\{ c(k,t) \}_{0 \leq k < \infty, 0 \leq t \leq T}$ on $X$ such that $c(k,t)= 0$ for $k \geq 2x_0$ and
\begin{equation}
\label{eq:cdef1.2}
 c(k,t) = \left\{ \begin{array}{ll}
                      (x_0 - k) + \int_0^t ds \int_{x_0 - k}^{x_0} dy (k+y - x_0) \frac{q(y,s)}{2} \hspace{5mm} & k < x_0 \\
                      \int_0^t ds \int_{k-x_0}^{x_0} dy (x_0 + y - k) \frac{q(y,s)}{2} \hspace{5mm} &  x_0 \leq k < 2x_0
                      \end{array} \right. .
\end{equation}
Here $q:[0,x_0] \times [0,T]$ is a density on $\R^+ \times [0,T]$, ie. $q \geq 0$ and $\int_0^T dt \int_0^{x_0} dy q(y,t) = 1$.

These call prices are consistent with a class of models in which the price process is constant except for a jump at the random time $\Sigma$. The joint law of the jump size and jump time has density $q$, and conditional on a jump of size $y$ at time $\sigma$ the jump is upwards with probability 1/2 and downwards with probability 1/2. We can write $X_t = x_0 + Z Y I_{ \{ t \geq \Sigma \} }$ where $(Y, \Sigma)$ has density $q$ and $Z$ is a Uniform random variable on $\{ \pm 1 \}$ which is independent of $Y$ and $\Sigma$.

Note that this remains a class of models since we have not yet specified when information about the pair $(Y,\Sigma)$ is revealed. It could be that $(Y,\Sigma,Z)$ is revealed at time $\Sigma$ (the `adapted' model) or it could be that $(Y,\Sigma)$ is known at time $0+$ and only $Z$ is revealed at time $\Sigma$ (the `foreknowledge' model). However there are many other models (filtered probability spaces supporting a martingale price process) consistent with the call prices in (\ref{eq:cdef1.2}). For instance, it is possible to define a generalised local volatility model for which call prices are given by (\ref{eq:cdef1.2}) --- generalised in the sense that there is an atom at $x_0$ at time $t$ of size $\int_t^T ds \int_0^{x_0} dy q(y,s)$, but away from $x_0$ the process acts like a martingale diffusion.

Our claim is that the highest model-based price for the American option consistent with the call price data in (\ref{eq:cdef1.2}) is from the `foreknowledge' model. In this model the values of $Y$ and $\Sigma$ are revealed at $t=0+$. Once they are known it is possible to either exercise immediately (with payoff $b(0+)=b(0))$ or exercise at $\Sigma$ with payoff $Y + b(\Sigma)$, and it is optimal to choose whichever of these exercise times yields the highest payoff.
The model-based price is then
\[ \phi = \int_0^T ds \int_0^{x_0} dy q(y,s) \max \{ b(0) , y+ b(s) \} = b(0) + \int_0^T ds \int_0^{x_0} dy q(y,s) ( y+ b(s) - b(0))^+. \]
Implicit in this model is a regime process which jumps to 2 at $t=0$ if $Y + b(\Sigma)<b(0)$ and otherwise jumps to $2$ at time $\Sigma$.

The optimality of the model will follow if we can exhibit a super-replicating strategy which has cost $\phi$. Consider the following (continuous time) semi-static strategy (assuming a price path $x \in D_{x_0}([0,T])$, and exercise at $\rho$) \\
$\bullet$
purchase a density (in time and space) of Arrow-Debreu style European payoffs which pay $- \dot{b}(s) I_{ \{ |x - x_0| > b(0) - b(s) \} }$ if $X_t =x$; \\
$\bullet$
in addition purchase a time $T$ terminal payoff $b(0) + [|x - x_0| - (b(0) - b(T)) ]^+$; \\
$\bullet$
take a short position in the asset if $x_{s-}>x_0 + b(0) - b(T)$ and $s \geq \rho$ and a long position if $x_{s-}<x_0 - b(0) + b(T)$ and $s \geq\rho$.

Observe that under a consistent model the time-$t$ density of the law of the price process at $x_0 \pm y$ (away from $x_0$) is $\frac{1}{2} \int_0^t q(s,y) ds$. Then the cost of the candidate super-hedging strategy is
\begin{eqnarray*}
\lefteqn{\int_0^T ds \int_0^{x_0} dy |\dot{b}(t)| I_{ \{ y > (b(0) - b(t))\} } \int_0^t q(s,y) ds + \int_0^T ds \int_0^{x_0} dy q(y,s) \{ b(0) + [ y - (b(0) - b(T)) ]^+ \} } \\
& = & b(0) + \int_0^T ds \int_0^{x_0} dy q(y,s) \left( [ y - (b(0) - b(T)) ]^+ - \int_s^T dt \dot{b}(t) I_{ \{ y > b(0) - b(t) \} } \right) \\
& = & b(0) + \int_0^T ds \int_0^{x_0} dy q(y,s)[ y - (b(0) - b(s)) ]^+ = \phi
\end{eqnarray*}

It remains to show that the strategy super-replicates. For $y \in D_{x_0}([0,T])$ let $\sG_T(y,\rho)$ denote the final payoff of the semi-static strategy.
Then, using $z_1(s) = y(s)+b(s) - (x_0 + b(0))$, $z_2(s) = b(s)-y(s) - (b(0) - x_0)$ and Lemma~\ref{lem:pathineq},
\begin{eqnarray*}
\sG_T(y,\rho) & = & b(0) + [ |y(T) - x_0| - (b(0)- b(T))]^+ + \int_0^T |\dot{b}(s)| I_{ \{ |y(s)-x_0| > b(0) - b(s) \} } ds \\
  && \hspace{5mm} - \int_{(\rho,T]} dy(s) I_{ \{ y(s-) > x_0 + b(0) - b(s) \} } + \int_{(\rho,T]} dy(s)  I_{ \{ y(s-) < x_0 - b(0) + b(s) \}} \\
& \geq & b(0) + [ |y(T) - x_0| - (b(0)- b(T))]^+ - \int_{(\rho,T]} \dot{b}(s) I_{ \{ |y(s-)-x_0| > b(0) - b(s) \} } ds \\
&& \hspace{5mm} - \int_{(\rho,T]} dy(s) I_{ \{ y(s-) > x_0 + b(0) - b(s) \} } + \int_{(\rho,T]} dy(s)  I_{ \{ y(s-) < x_0 - b(0) + b(s) \}} \\
& = & b(0) + z_1(T)^+ + z_2(T)^+ - \int_{(\rho,T]} dz_1(s) I_{ \{ z_1(s-) > 0 \} } - \int_{(\rho,T]} dz_2(s) I_{ \{ z_2(s-) > 0 \} } \\
& \geq & b(0) + z_1(\rho)^+ + z_2(\rho)^+ \\
& = & b(0) + [ |y(\rho) - x_0| - (b(0)- b(\rho))]^+ \geq  |y(\rho) - x_0| + b(\rho) = a(y(\rho),\rho)
\end{eqnarray*}
where we use $z_1(t)^+ + z_2(t)^+ =  [|y(t) - x_0| - (b(0) - b(t))]^+$. Hence the semi-static strategy super-replicates.

We have shown that there is a consistent model and a semi-static super-replicating strategy for which the model price and the cost of the super-replicating strategy coincide. Hence we have found the highest model price and the cheapest super-replicating strategy.

\section{Conclusions}
\label{sec:conclusion}

To gain insight into the potential value of an American claim in the presence of a set of closely related hedging instruments (European options on the same underlying) this paper develops a method of computing the maximum possible value of the claim. The main message of the paper is that much of the value of the American option arises from model uncertainty, and the ability of the holder of the American claim to adapt his strategy as that uncertainty is resolved, a possibility which is not available to the holder of a European option. It is not possible to capture an evolution of beliefs about the distribution of future returns in models in which the filtration is the natural filtration of the price process, and in order to capture the full value of the American option it is necessary to work with more general probabilistic set-ups. Approaches to pricing which place strong assumptions on the flow of information implicitly place severe constraints on the future prices of the instruments used for hedging, leading to an underestimate of the full value of the American claim and an exaggeration of the efficacy of hedging strategies.

The specific analysis in the paper is for a discrete-time, discrete-space universe. In that setting the approach can readily be implemented as a linear program. We can describe the model for which the model price is maximised (over the class of models consistent with the European prices). The model is a two-regime model in which the option is exercised at the moment when the regime changes. We can also describe the cheapest super-replicating strategy. This strategy is a semi-static strategy in which the dynamic hedge in the underlying depends only upon the price level of the underlying, and whether or not the American option has been exercised. There is no duality gap: the most expensive model price is equal to the cost of the cheapest super-replicating strategy.

Several extensions of the results are possible under weak additional assumptions. We can allow the price process to take values in $\R^+$, and for exercise to occur at any time (and not just the times which correspond to the maturities of the European options). However, an assumption throughout is that the set of traded options is finite. It is an interesting question to ask if the methods of this paper can be generalised to the setting of an infinite number of options.
As evidenced by our example in the penultimate section, we believe that the discrete framework we describe captures the essential features of the problem, and that the main message will be unaltered in a more general setting. There will be major challenges however in determining the most appropriate definitions for the various notions involved, especially in continuous time.

This paper has practical application in hedging American claims. Standard hedging techniques use options as well as the underlying asset to maintain a position that has no exposure to the `greeks': delta, gamma and so on. But hedging in this way faces two problems: it works badly if the model is mis-specified, since when the portfolio is rebalanced, the hedging securities are necessarily traded at market prices which may differ substantially from model prices; and trading options, as required with dynamic gamma-hedging for example, tends to incur substantial transaction costs. By contrast, the semi-static strategy used in the robust pricing literature works however the world behaves; the strategy puts a firm floor on the maximum loss that can be incurred in any state of the world. Furthermore, after time 0, the strategy requires trading only in the underlying asset where transaction costs are generally far lower than they are for options.

The large potential for mis-valuation of the American option suggests that the search for ever more accurate and rapid computational procedures for evaluating the early exercise premium needs to be tempered by an awareness of the sensitivity of the results to the particular model of price dynamics that is being used. The point is likely to be particularly significant in the presence of event risk (as in battles for corporate control, or currencies under speculative attack) where there are several scenarios, with different implications for future price volatility, whose probabilities vary substantially over time.

\end{document}